\documentclass[11pt]{article}
\usepackage{subfigure}

\usepackage[letterpaper, portrait, margin=1in]{geometry}
\usepackage{algorithm}
\usepackage[noend]{algpseudocode}
\usepackage{url}
\usepackage{amsmath,amssymb,amsthm}
\usepackage{mathtools}
\usepackage{parskip}
\usepackage{xspace}
\usepackage{tikz}
\usepackage{dsfont}

\theoremstyle{plain}
\newtheorem{theorem}{Theorem}[section]
\newtheorem{lemma}[theorem]{Lemma}

\newtheorem{proposition}[theorem]{Proposition}
\newtheorem{claim}[theorem]{Claim}

\theoremstyle{definition}
\newtheorem{definition}[theorem]{Definition}
\newtheorem{remark}[theorem]{Remark}
\newtheorem{example}[theorem]{Example}

\newcommand{\ignore}[1]{}

\newcommand{\dist}{\mathsf{dist}}
\newcommand{\hdist}{\mathsf{hdist}}
\newcommand{\acos}{\cos^{-1}}
\newcommand{\asin}{\sin^{-1}}

\usepackage{wrapfig}
\usepackage{hyperref}
\DeclareMathOperator{\Cov}{Cov}
\newcommand{\Var}[1]{\mathrm{Var} \left[ #1 \right]}
\newcommand{\Ex}[1]{\bE \left[ #1 \right]}
\newcommand{\Exu}[2]{\underset{#1} \bE \left[ #2 \right] }
\newcommand{\Exuc}[3]{\underset{#1} \bE \left[ #2 \;\; \left| \;\; #3
\right.\right] }

\renewcommand{\Pr}[1]{\bP \left[ #1 \right]} 
\newcommand{\Pru}[2]{\underset{ #1 }\bP \left[ #2 \right]}
\newcommand{\Pruc}[3]{\underset{ #1 }\bP \left[ #2 \;\; \mid \;\; #3 \right]}

\newcommand{\BigO}[1]{O \left( #1 \right) }

\newcommand{\BigOmega}[1]{ \Omega \left( #1 \right) }

\newcommand{\BigTheta}[1]{\Theta \left( #1 \right) }
\newcommand{\define}{\vcentcolon=}

\renewcommand{\exp}[1]{\mathrm{exp}\left( #1 \right)}

\DeclareMathOperator{\sign}{sign}
\DeclareMathOperator{\poly}{poly}
\newcommand{\xh}[1]{\hspace{2em}\text{#1}}

\newcommand{\ceil}[1]{\ensuremath{\lceil #1 \rceil}}
\newcommand{\norm}[1]{\| #1 \|}
\newcommand{\fracb}[2]{\left( \frac{#1}{#2} \right)}
\newcommand{\inn}[1]{\langle #1 \rangle}
\newcommand{\abs}[1]{| #1 |}
\newcommand{\ind}[1]{\mathds{1} \left[ #1 \right] }
\newcommand{\nth}[1]{\ensuremath{#1^{\textup{th}}}}

\newcommand{\pmset}{\{\pm 1\}}

\newcommand{\cN}{\ensuremath{\mathcal{N}}}

\newcommand{\cX}{\ensuremath{\mathcal{X}}}

\newcommand{\bE}{\ensuremath{\mathbb{E}}}

\newcommand{\bP}{\ensuremath{\mathbb{P}}}
\newcommand{\bR}{\ensuremath{\mathbb{R}}}

\title{Testing Halfspaces over Rotation-Invariant Distributions}
\date{}
\author{Nathaniel Harms\thanks{This research was partially supported by the
Natural Sciences and Engineering Research Council of Canada.} \\ University of Waterloo \\
\texttt{nharms@uwaterloo.ca}}
\begin{document}
\maketitle
\thispagestyle{empty}

\begin{abstract}
We present an algorithm for testing halfspaces over arbitrary, unknown
rotation-invariant distributions.  Using $\widetilde O(\sqrt n \epsilon^{-7})$
random examples of an unknown function $f$, the algorithm determines with high
probability whether $f$ is of the form $f(x) = \sign(\sum_i w_i x_i -t)$ or is
$\epsilon$-far from all such functions. This sample size is significantly
smaller than the well-known requirement of $\Omega(n)$ samples for learning
halfspaces, and known lower bounds imply that our sample size is optimal (in its
dependence on $n$) up to logarithmic factors. The algorithm is distribution-free
in the sense that it requires no knowledge of the distribution aside from the
promise of rotation invariance. To prove the correctness of this algorithm we
present a theorem relating the distance between a function and a halfspace to
the distance between their centers of mass, that applies to arbitrary
distributions.

\end{abstract}

\newpage
\clearpage
\setcounter{page}{1}

\section{Introduction}
Halfspaces (or linear threshold functions) cut $\bR^n$ in half by drawing a
hyperplane through the space. They are defined by a vector $w \in \bR^n$ and a
threshold $t$ and they label points $x \in \bR^n$ by $\sign(\sum_i w_i x_i -
t)$. These functions are fundamental objects in many areas of study, like
machine learning, geometry, and optimization.  In machine learning, halfspaces
are one of the most basic classes of functions that can be learned and
algorithms for doing so have been studied extensively: it is well-known that, in
the PAC learning model, a \emph{distribution-free} algorithm (one with no
knowledge of the distribution of examples) requires $\BigTheta{n/\epsilon}$
random samples to learn a halfspace with accuracy $\epsilon$ \cite{SB14};
however, this requires that the function is guaranteed to be a halfspace, and in
practice it may not be known what class the function belongs to. In this
situation we can try to learn a halfspace and then check how well it works, but
it would be better to find an algorithm that could quickly reject functions that
are not halfspaces, which is the goal of the present paper. This is the problem
of \emph{testing halfspaces}: an $\epsilon$-testing algorithm for halfspaces is
an algorithm which receives random samples of a function $f$, and determines
whether the function is a halfspace or \emph{$\epsilon$-far} from all
halfspaces, where distance is measured by the probability that two functions
differ on a random input.

Matulef \emph{et al.}\cite{MORS10} and Mossel and Neeman \cite{MN12} give
algorithms for testing halfspaces when the algorithm is allowed to make queries
to the function and the underlying distribution of points is either Gaussian
over $\bR^n$ or uniform over $\pmset^n$. Balcan \emph{et al.} \cite{BBBY12}
present an algorithm in the model where the tester receives random samples from
the Gaussian distribution. It is unknown whether there is a distribution-free
testing algorithm, in the query or sample model, that improves upon the na\"ive
testing-by-learning strategy. Existing algorithms are tailored for the Gaussian
or the hypercube, so for distributions significantly different from these, even
for another rotation-invariant (RI) distribution, a testing algorithm is not
known.  RI distributions generalize both the Gaussian distribution and the
uniform distribution over the sphere, two distributions that are commonly
studied in the literature on halfspaces (e.g.~\cite{BBBY12, Kearns98, Long94, Long03,
MORS10, MN12}). In this paper we present an algorithm that will work for
any \emph{unknown} rotation-invariant distribution. 
\begin{theorem}
\label{thm:ri tester}
There is an algorithm that is, for any unknown rotation-invariant distribution
$\mu$ and any $\epsilon > 0$, an $\epsilon$-tester for halfspaces using at
most $\widetilde O(\sqrt n \epsilon^{-7})$ \footnote{The notation $f(n) =
\widetilde O(g(n))$ hides logarithmic factors; i.e.~for sufficiently large $n,
f(n) \leq g(n) \log^c n $ for some constant $c$, and $f(n) = \widetilde
\Omega(g(n))$ means $f(n) \geq g(n) \log^{-c} n$ for sufficiently large $n$.}
random samples.
\end{theorem}

This algorithm requires significantly fewer samples than the
$\Omega(n/\epsilon)$ required for the na\"ive testing algorithm that attempts to
learn the function \cite{GGR98}, and in fact it is optimal (up to logarithmic
factors) due to the $\widetilde \Omega(\sqrt n)$ lower bound 
for algorithms sampling from the Gaussian distribution \cite{BBBY12}.

Rotation invariance and halfspaces go hand-in-hand since, in an RI distribution,
the normal vector $w$ of a halfspaces is parallel to its center of mass, which,
for example, allows one to learn a halfspace by learning its center
(e.g.~\cite{Kearns98}).
\begin{definition}[Center of Mass]
Let $\mu$ be a distribution over $\bR^n$ and let $f : \bR^n \to \pmset$ be a
measurable function. The \emph{center of mass} of $f$ is the vector
$\bE_x[xf(x)]$ and the \emph{center-norm} is $\| \bE_{x}[xf(x)] \|_2$.
\end{definition}
To prove the correctness of our algorithm, we show that for arbitrary
distributions, the distance between a function and a halfspace is related to the
distance between the two centers of mass, where this relationship is quantified
by a parameter we call the \emph{width}, $W_\mu(w,\epsilon)$, which is,
informally, the size of the smallest interval $I$ with $\bP_x[\inn{w,x} \in I] =
\epsilon$ (see Definition \ref{def:width}).
\begin{theorem}[See Section \ref{section:centers of mass} for the full
statement]
\label{thm:gap theorem}
For any distribution $\mu$ over $\bR^n$ and any function $f$, if $h$ is a
halfspace with normal vector $w$ having distance $\epsilon$ to $f$ and the same
mean as $f$, then
\[
  \norm{\Ex{xh(x)}-\Ex{xf(x)}}_2 \geq \epsilon \cdot W_\mu(w,\epsilon/2) \,.
\]
\end{theorem}
Although the algorithm itself is for RI distributions, Theorem \ref{thm:gap
theorem} actually holds much more generally.  When applied to the Gaussian
distribution, we recover a result of \cite{MORS10} that was used to prove the
correctness of their algorithm. That result used Fourier analysis in the
Gaussian space, which we eliminate in our proof.

The testing algorithms of \cite{MORS10,BBBY12} work by estimating the
center-norm of the input function.  In this paper we show that these algorithms
do not generalize to the larger class of RI distributions.
\begin{theorem}[Informal]
\label{thm:counterexample}
For any algorithm $A$ that has access to a function $f$ only through an
estimate of its center-norm, there is an RI distribution and a function
$f$, far from all halfspaces, that is indistinguishable to $A$ from a halfspace.
\end{theorem}

\ignore{
It may not be immediately
obvious that any work is required to generalize the algorithms for the Gaussian
space to arbitrary rotation-invariant distributions, so we give an example that
shows that the algorithms of \cite{BBBY12, MORS10} are not sufficient. These
algorithms work by estimating the center-norm of the function, and we show that
an algorithm of this type will not work.

\begin{theorem}[Informal]
\label{thm:counterexample}
For any algorithm
$A$ that has access to a balanced function $f$ only through an estimate of the
center-norm, there exists an RI distribution and a balanced function $f$,
$1/4$-far from all halfspaces, that is indistinguishable to $A$ from a
halfspace.
\end{theorem}
}

\subsection{Related Work}

Matulef \emph{et al.} \cite{MORS10} presented algorithms that use queries to test
halfspaces over the Gaussian distribution and the uniform distribution over
$\pmset^n$ with query complexity $\poly(1/\epsilon)$. Mossel and Neeman give
another $\poly(1/\epsilon)$-query algorithm for the Gaussian distribution that
estimates the noise stability instead of the center-norm.  Glasner and Servedio
\cite{GS07} show a $\widetilde \Omega(n^{1/5})$ lower bound for
distribution-free testing in the query model, which contrasts with the
$\poly(1/\epsilon)$ algorithms to show that testing with queries is much harder
for general distributions than for the Gaussian.  Balcan \emph{et al.}
\cite{BBBY12} adapt the center-norm algorithm of \cite{MORS10} to get a sampling
algorithm for the Gaussian distribution that requires $O(\sqrt{n \log n})$
random samples, and also show a lower bound of $\Omega(\sqrt{n / \log n})$ in
the same setting.

There are a few related works that study similar problems: Matulef \emph{et al.}
\cite{MORS09} and Ron and Servedio \cite{RS15} gave algorithms for testing,
under the uniform distribution over $\pmset^n$, if a function is a halfspace
with normal vector $w \in \pmset^n$ (as opposed to $w \in \bR^n$), which they
show requires between $\Omega(\log n)$ and $\poly(\log n)$ queries, rather than
the constant number of queries required for testing halfspaces.  Balcan \emph{et
al.} \cite{BBBY12} also give a lower bound of $\Omega((n / \log n)^{1/3})$
labelled examples in their \emph{active testing} model, where the algorithm
receives unlabelled, randomly selected points, and may request labels for a
subset of these. Finally, Raskhodnikova \cite{Ras03} presents a query algorithm
for testing halfspaces when the domain of the function is an image, i.e.~an $n
\times n$ matrix of binary values, that uses $O(1/\epsilon)$ queries.

Learning halfspaces is a related and well-studied problem; standard arguments
about the VC dimension of halfspaces show that $\BigTheta{\frac{n}{\epsilon}}$
random samples or queries are necessary and sufficient to learn halfspaces in
the PAC learning model \cite{SB14}; this holds when the algorithm is required to
be \emph{distribution-free} (i.e.~it is not given any knowledge of the
distribution), but the lower bound of $\Omega(n/\epsilon)$ holds even when the
distribution is uniform over the unit sphere in $\bR^n$
\cite{Long94}. If we imagine that a testing algorithm has accepted a
function $f$, then we know it is of distance at most $\epsilon$ to a halfspace;
if we now wish to learn the function, the standard PAC model does not directly
apply since the function still may not be exactly a halfspace. The learning
model that applies in this situation is \emph{agnostic learning}, in which the
labels are not guaranteed to be consistent with a halfspace. In this model, the
VC dimension arguments show that $\Theta(n/\epsilon^2)$ random examples are
required to produce a nearly-optimal halfspace \cite{SB14}, but finding this
halfspace is NP-Hard \cite{BEL03}; if we allow the learner to produce a function
that is not necessarily a halfspace, learning can be done in polynomial time
with $n^{O(1/\epsilon^4)}$ random examples when the distribution is uniform
over the sphere \cite{KKMS08}.

Some recent work has been done on the \emph{Chow parameters problem}
\cite{DDFS14, OS11}, which asks how to construct a halfspace over $\pmset^n$
given approximations of its center of mass and mean (known as \emph{Chow
parameters}). That work uses similar ideas to the work on testing and learning
halfspaces, such as bounds on the center-norms.

\section{Outline \& Sketch of the Algorithm}
In Section 3 we cover the notation and preliminaries required to present the
algorithm, including some important facts about centers of mass,
rotation-invariant distributions, and distances between functions.

The remainder of the paper is a construction of the testing algorithm,
\textsc{RI-Tester}, found in \S \ref{section:ri tester}, which we now sketch.
Our first task is to show that the known algorithms (\cite{BBBY12, MORS10}) for
testing halfspaces fail on general RI distributions. This is
done in \S \ref{subsection:width} by defining the class of \emph{center-of-mass
algorithms} and constructing RI distributions to fool them (Theorem
\ref{thm:counterexample}). This justifies our use of a new strategy and
motivates the definition of the \emph{width} and \emph{bounded RI distributions}
(Definitions \ref{def:width} and \ref{def:bounded distributions}) to isolate
those distributions on which center-of-mass algorithms fail.

We then use the definition of \emph{width} to prove a bound on the distance
between two functions in terms of the distance between their centers of mass
(Theorem \ref{thm:gap theorem}). Using this theorem, we show that there is an
algorithm, \textsc{Simple-Tester}, which will test halfspaces over bounded RI
distributions, and has a small amount of tolerance, i.e.~it will accept
halfspaces and anything sufficiently close to a halfspace. The idea of this
algorithm is the same as \cite{BBBY12, MORS10}: estimate the center-norm and
compare it with the center-norm we would expect if the function were a
halfspace, but our algorithm eliminates the dependence on the Gaussian
distribution. This method relies on the property that the center-norm of a
halfspace is independent of its orientation in RI spaces.

Next we reduce the general RI distribution problem to the bounded case. Note
that if $f$ is a halfspace with threshold $t$, then $f$ is constant on the ball
of radius $|t|$ and it is nearly balanced on points $x$ with $\|x\| \gg t$.  The
subroutine \textsc{Find-Pivot} in \S \ref{subsection:find pivot} identifies
these extreme regions by examining a set of examples and finding the smallest
possible threshold. These extreme regions are treated separately from the rest
of the space.

The ball of radius $|t|$ must be constant, which is guaranteed implicitly by the
\textsc{Find-Pivot} algorithm. Balanced halfspaces are approximately
preserved by rescaling the space (Lemma \ref{lemma:projecting balanced
functions}).  Since a halfspace $f$ is nearly balanced in the extreme region of
radius $\gg t$, we may rescale this region so that it is bounded, while
approximately preserving the halfspace property. We apply \textsc{Simple-Tester}
to this region, which requires the tolerance guarantee.

Between the two extreme regions we partition the space into bounded sections.
Now we apply \textsc{Simple-Tester} to each of these bounded sections along with
the transformed outer region. Now each region on its own is a halfspace, so what
remains is to ensure that these halfspaces are consistent with one another. This
is accomplished with the \textsc{Check-Consistency} subroutine in \S
\ref{subsection:check consistency} that checks if two (near-)halfspaces on two
different RI distributions are consistent, i.e.~that they have roughly the same
thresholds and orientation.  This check again depends on rotation invariance,
since in RI spaces the orientation of the halfspaces does not change their
center-norms, and the centers of mass are parallel to their normal vectors.

It may be the case that some sections did not contain enough random samples to
run the \textsc{Simple-Tester} and \textsc{Check-Consistency} algorithms. This
is handled by Lemma \ref{lemma:ignore small sections}, which states that any
section with insufficient samples also has small measure and may be ignored.

\section{Preliminaries and Notation}
\label{section:preliminaries}
See the appendix for the proofs of the propositions in this section.

For a vector $u \in \bR^n$, $\norm{u}$ is the 2-norm $\sqrt{\sum_i u_i^2}$. For
two vectors $u,v \in \bR^n$ we will write $\inn{u,v}$ for the inner product
$\sum_i u_i v_i$.  $\bP, \bE$ denote the probability of an event and the
expected value respectively. We will write $e_1$ for the first standard basis
vector $(1, 0, \dotsc, 0)$.

\begin{definition}[Halfspaces]
A \emph{halfspace} is a function $h : \bR^n \to \pmset$ such that for some unit
vector $w \in \bR^n$ (called the \emph{normal vector}) and threshold $t \in
\bR$
\[
  h(x) = \sign(\inn{w,x}-t) \,.
\]
\end{definition}
We use the standard definition of distance between functions.
\begin{definition}[Distance]
Let $\mu$ be a distribution over $\bR^n$ and let $f,g : \bR^n \to \pmset$ be
measurable functions. Then the distance is defined as
\[
  \dist_\mu(f,g) \define \Pru{x \sim \mu}{f(x) \neq g(x)} \,.
\]
We say $f$ is \emph{$\epsilon$-far} from $g$ if $\dist_\mu(f,g) \geq \epsilon$
and \emph{$\epsilon$-close} if $\dist_\mu(f,g) \leq \epsilon$.
\end{definition}
Functions that are close to being the same constant function are close to each
other.
\begin{proposition}
\label{prop:distance from mean}
Let $f,g$ be $\pm 1$-valued functions with $\abs{\Ex f} \geq 1-\epsilon$ and
$\abs{\Ex f - \Ex g} \leq \delta$. Then $\dist(f,g) \leq \epsilon + \delta/2$.
\end{proposition}

\subsection{Spheres and Rotation Invariance}
\label{subsection:spheres and rotation invariance}
\begin{definition}[Rotation-Invariant Distributions]
\label{def:ri distributions}
  Write $\sigma$ for the uniform distribution over the unit sphere.
  A distribution $\mu$ over $\bR^n$ is \emph{rotation-invariant} (RI) if for
  some distribution $\mu_\circ$ over $\bR_{\geq 0}$, $\mu = \sigma \mu_\circ$.
  That is, for a random vector $x \sim \mu$, $x = au$ for independent random
  variables $u \sim \sigma$ and $a \sim \mu_\circ$. We will refer to $\mu_\circ$
  as the \emph{radial distribution} of $\mu$.
\end{definition}
We will frequently use the projection of an RI distribution onto a
one-dimensional line: for RI distributions this projection is the same for every
line. We will define some notation for this projection operation:
\begin{definition}[1-Dimensional Projection]
Let $\mu$ be any RI distribution over $\bR^n$ with radial distribution
$\mu_\circ$. We define $\mu_\pi$ to be the distribution of $\inn{x,e_1}$ where
$x \sim \mu$ and $e_1$ is the first standard basis vector. Note that by rotation
invariance, we could define $\mu_\pi$ as the distribution of $\inn{x,u}$ for any
unit vector $u \in \bR^n$.
\end{definition}
We will usually assume that our distributions have the same scale. For RI
distributions it is convenient to assume that they are \emph{isotropic}, which
means that the 1-dimensional projection has unit variance.
\begin{definition}[Isotropic]
A distribution $\mu$ on $\bR^n$ is \emph{isotropic} if for every unit vector $u
\in \bR^n, \bE[\inn{u,x}^2]=1$.
\end{definition}
\begin{proposition}
\label{prop:isotropy}
Let $\mu$ be any RI distribution over $\bR^n$. $\mu$ is isotropic iff $\bE_{x
\sim \mu}[\norm{x}^2]=n$.
\end{proposition}

RI distributions are essentially convex combinations of spheres, and the uniform
distribution over the sphere behaves similarly to the Gaussian distribution.
This is evident in the following propositions giving bounds on the tails and on
the density of the 1-dimensional projection. The proofs  of these propositions
can be found in the appendix.
\begin{proposition}
\label{prop:tail bounds for spheres}
Let $\sigma$ be the uniform distribution over the sphere of radius $r$, and let
$t \geq 0$. Then for any unit vector $u$,
\[
\Pru{x \sim \sigma}{\inn{u,x} \geq t} \leq \sqrt{2}e^{-t^2(n-2)/2r^2} 
\leq \sqrt{2}e^{-t^2n/4r^2}
\]
(where the last inequality holds when $n \geq 4$).
\end{proposition}
\begin{proposition}
\label{prop:density of sphere}
Let $\sigma$ be the uniform distribution of the sphere of radius $r$ over
$\bR^n$. Let $\sigma_\pi$ be the density of the 1-dimensional projection.
Then $\sigma_\pi(x) \leq \frac{\sqrt{n-1}}{\sqrt{2\pi r^2}}
e^{-\frac{x^2(n-2)}{2r^2}}$, and in particular, for $r = \sqrt{n}$,
$\sigma_\pi(x) \leq \frac{1}{\sqrt{2\pi}} e^{-x^2/4}$ when $n \geq 4$.
\end{proposition}
\ignore{\begin{proposition}
\label{prop:projection variance of unit sphere}
  Let $x$ be drawn uniformly at random from the sphere with radius $r$ and let
  $u$ be any unit vector. Then $\Ex{\inn{x,u}^2} \leq \frac{\sqrt{3}
  r^2}{\sqrt{2}(n-2)}$.
\end{proposition}}
\ignore{
\begin{proposition}
\label{prop:sphere absolute value}
Let $\mu$ be the uniform distribution over the sphere with radius $r$ and
dimension $n \geq 4$. Then for any unit vector $u$, $\Ex{\abs{\inn{x,u}}} \in
\left[\frac{r}{2\sqrt{n}}, \frac{r}{\sqrt{n}} \right]$.
\end{proposition}
}
\ignore{
By decomposing an RI distribution into its component spheres we can get a bound
on the 1-dimensional variance:
\begin{proposition}
\label{prop:projection variance is constant}
Let $\mu$ be any isotropic RI distribution. Then for any unit vector $u$, the
variance of the one-dimensional projection satisfies $\Ex{\inn{x,u}^2} \leq
\sqrt{6}$.
\end{proposition}
\begin{proof}
Let $\mu_\circ$ be the radial distribution of $\mu$. Then
\begin{align*}
    \Exu{x \sim \mu}{\inn{x, u}^2}
    &= \Exu{r \sim \mu_\circ}{\Exuc{x \sim \mu}{\inn{x, u}^2}{\norm{x}=r}}
    = \Exu{r \sim \mu_\circ}{r^2 \Exuc{x \sim \mu}{\inn{x,  u}^2}{\norm{x}=1}} \\
    &= \Ex{\norm{x}^2} \Exuc{x \sim \mu}{\inn{x,  u}^2}{\norm{x}=1}
    \leq n \frac{\sqrt{3}}{\sqrt{2}(n-2)} \leq \sqrt{6}
\end{align*}
where the penultimate inequality is due to Proposition \ref{prop:projection
variance of unit sphere} and the final inequality is $n-2 \geq n/2$ for $n \geq
4$.
\end{proof}
}
\subsection{Halfspaces and Rotation Invariance}
\label{subsection:halfspaces and rotation invariance}
Halfspaces have several useful properties in RI distributions, notably the fact
that the normal vectors and centers of mass are parallel and that the distance
between two halfspaces can be decomposed into angular and threshold components,
which we present here.
\begin{proposition}
\label{prop:center and normal are parallel}
For any RI distribution $\mu$ over $\bR^n$ and any halfspace $h(x) =
\sign(\inn{w,x}-t)$, there exists a scalar $s > 0$, such that $\bE_{x \sim
\mu}[x h(x)] = s w$.
\end{proposition}

We call two halfspaces \emph{aligned} if their normal vectors (and therefore
centers of mass) are parallel. For arbitrary functions, alignment refers to the
centers of mass:
\begin{definition}[Aligned functions]
Let $\mu$ be any probability distribution over $\bR^n$ and
let $f, g : \bR^n \to \pmset$ be measurable functions. We will say $f$ and $g$
are \emph{aligned} (with respect to $\mu$) if their centers of mass are
parallel, i.e.~for some scalar $s > 0$, $\bE_{x \sim \mu}[x f(x)] = s \bE_{x \sim
\mu}[x g(x)]$. In particular, if $f,g$ are halfspaces with normal vectors $u,v$,
and $\mu$ is RI, then $f,g$ are aligned if and only if $u = sv$.
\end{definition}

Now we present the decomposition of distance into the angular component and the
threshold component; following are the definitions of these components.
\begin{definition}[Threshold Distance]
For any RI distribution $\mu$ we will we will define $\hdist_\mu$ as a metric on
$\bR$ as follows: for $a, b \in \bR$ define the halfspaces $h_a(x) = \sign(x_1 -
a), h_b(x) = \sign(x_1 - b)$. Then
\[
  \hdist_\mu(a,b) = \dist_\mu(h_a, h_b) \,.
\]
\end{definition}
We will drop the subscript when the distribution is clear from context. Next we
will introduce notation for the angle between two halfspaces:
\begin{definition}[Halfspace Angle]
For two halfspaces $g, h$ with normal (unit) vectors $u,v$ respectively, we will
write $\alpha(g,h)$ for the angle between $u$ and $v$:
\[
\alpha(g,h)
\define \acos(\inn{u,v}) \,.
\]
\end{definition}

We can decompose the distance between two halfspaces into the sum of these two
metrics:

\begin{proposition}
\label{prop:distance decomposition}
  For any two halfspaces $h(x) = \sign(\inn{v,x}-q), g(x) = \sign(\inn{u,x}-p)$
  (where $u,v$ are unit vectors) and RI distribution $\mu$,
  \[
    \dist_\mu(h,g) \leq \frac{\alpha(h,g)}{\pi} + \hdist_\mu(p,q) \,.
  \]
\end{proposition}

\ignore{
Suppose there exist points $x$ on the sphere such that $g(x)=h'(x)=1$. Then
$\inn{u,x}, \inn{v,x} \geq p$
Suppose $\inn{b,v} \geq p$. 

We can easily see that $\alpha(a,v) =
\alpha(b,u)$ which implies $\alpha(a,b) = \alpha(u,v) = \alpha$. Thus in case 1,
where lines $ac, bd$ intersect inside the circle, we have $\dist_{\sigma_r}(g,h')
= \frac{2\alpha(a,b)}{2\pi} = \alpha/\pi$. And in case 2, where the lines
intersect outside the circle, we have $\dist_{\sigma_r}(g,h') =
\frac{2\alpha(a,c)}{2\pi} \leq \alpha(a,b)/\pi = \alpha/\pi$.
}
Finally, here is an identity for computing the mean of halfspaces:
\begin{proposition}
  \label{prop:halfspace mean}
  Let $h$ be a halfspace with threshold $t$ and normal vector $w$. Then for
  any RI distribution,
  \[
  \abs{\Ex h} = \Pr{-|t| \leq \inn{w,x} \leq |t|} \,.
  \]
\end{proposition}

\subsection{Concentration Inequalities and Empirical Estimation}
We refer the reader to the appendix for a discussion of the standard empirical
estimation techniques used in this paper. For now it suffices to mention that
the algorithms in this paper will frequently refer to the algorithm
$\textsc{Estimate-Mean}(\mu, f, \epsilon, \delta)$ which produces an estimate of
$\bE_{x \sim \mu}[f(x)]$ with additive error $\epsilon$ and confidence $\delta$.

\section{Centers of Mass}
\label{section:centers of mass}
Our algorithms will rely primarily on the properties of the centers of mass of
Boolean functions, or more specifically the norms of these centers, which we
will call the \emph{center-norms}. Centers of mass are fundamental quantities in
the study of halfspaces: in the language of Fourier analysis of Boolean
functions (see for example \cite{OD14} for an introduction to this topic), what
we call the center of mass $\Ex{x f(x) }$ is the vector of \emph{degree-1
Fourier coefficients} (one can see this by observing that the $\nth i$ Fourier
coefficient is $\hat f(i) = \Ex{x_i f(x) }$). The relationship between the degree-1 Fourier
coefficients and halfspaces has arisen in a number of works (e.g. \cite{Chow61,
MORS10, OD14}). Another name for the set of $n+1$ quantities $\Ex f, \Ex{x_i
f(x)}$ is the \emph{Chow parameters}, after a result of Chow which states that
for any function $f : \pmset^n \to \pmset$ and any halfspace $h$, if $f$ and $h$
share the same Chow parameters then $f = h$ \cite{Chow61}. Much of the work on
halfspaces can be interpreted as variations on this theme, including our Theorem
\ref{thm:gap theorem}.

It has been known for some time that halfspaces maximize the quantity
$\norm{\Ex{xf(x)}}$ when the distribution is uniform over $\pmset^n$ or Gaussian
(e.g. \cite{Wind71}). The main result of this section is a generalization of
this relationship, Theorem \ref{thm:gap theorem}: we will show that for any
distribution $\mu$ over $\bR^n$ with bounded \emph{width}, halfspaces not only
maximize the center-norm but the larger the center-norm of a function, the
closer it must be to a halfspace. The other results in this section are bounds
on the center-norms that we will require to prove guarantees on our algorithms.

\subsection{Width}
\label{subsection:width}
Recall the example from the introduction (Theorem \ref{thm:counterexample})
which stated that pure center-of-mass algorithms cannot work as testers. We
now prove this theorem, which will illustrate that the main obstacle to
generalizing the center-of-mass algorithms is that RI distributions can be
densely concentrated on very small areas. This motivates the definition of
\emph{width}, which quantifies this dense concentration.

The algorithms of \cite{MORS10, BBBY12} work by making an estimate of the
center-norm: we will refer to such algorithms as \emph{center-of-mass
algorithms} and interpret them as making a single query to a center-norm
estimator. Such algorithms should work regardless of which estimator is used, as
long as it satisfies a sufficient accuracy guarantee. We formalize this in the
next definition, which is inspired by the Statistical Query model of Kearns
\cite{Kearns98}:
\begin{definition}
  A \emph{center-norm oracle} for an isotropic rotation-invariant
  distribution $\mu$ and function $f$ is an oracle $C_{f,\mu}(\epsilon)$ which
  on request $\epsilon \geq 0$ produces a random variable $C$ satisfying
  \[
  \bP[| C - \|\Exu{x \sim \mu}{xf(x)}\|_2 | > \epsilon] \leq 1/3 \,.
  \]
  A \emph{center-of-mass tester} for balanced halfspaces is an algorithm $A$
  which, when given access to any center-norm oracle $C_{f,\mu}$,
  must satisfy the following\footnote{Making
  more than one request to this oracle is superfluous: we could provide an oracle
  that always overestimates by exactly $\epsilon$, so making multiple requests
  would give no more information than just the request for highest accuracy.}:
  \begin{enumerate}
    \item If $f$ is a balanced halfspace, then $A(C_{f,\mu})$ accepts with
      probability at least $2/3$, and
    \item If $f$ is $1/4$-far (in $\mu$) from all halfspaces, then
      $A(C_{f,\mu})$ rejects with probability at least $2/3$.
  \end{enumerate}
  Observe that $A$ does not have any knowledge of $\mu$, and note that we are
  weakening the standard requirements for a tester: we are only required to
  accept when $f$ is balanced, and we are concerned with only constant distance.
  For the Gaussian distribution, the algorithms of \cite{MORS10, BBBY12} would
  satisfy this definition, using a request $\epsilon = \Omega(1)$.
\end{definition}
We can now show that such algorithms cannot exist for unknown isotropic RI
distributions, since without knowledge of the distribution we must demand
perfect estimations of the center-norm:

\textbf{Theorem \ref{thm:counterexample}.}\emph{
Any center-of-mass tester $A$ for balanced halfspaces must request accuracy
$\epsilon = 0$ from the center-norm oracle.
}

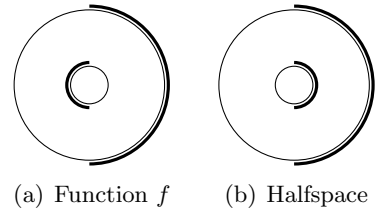
\begin{wrapfigure}[10]{r}{2.5in}
\centering
\mbox{
\subfigure[Function $f$]{
\begin{tikzpicture}
  \draw (0,0) circle [x radius=1, y radius=1];
  \draw (0,0) circle [x radius=0.25, y radius=0.25];
  \draw[very thick] (90:1.05) arc (90:-90:1.05) ;
  \draw[very thick] (90:0.3) arc (90:270:0.3) ;
\end{tikzpicture}
}\quad
\subfigure[Halfspace]{
\begin{tikzpicture}
  \draw (0,0) circle [x radius=1, y radius=1];
  \draw (0,0) circle [x radius=0.25, y radius=0.25];
  \draw[very thick] (90:1.05) arc (90:-90:1.05) ;
  \draw[very thick] (90:0.3) arc (90:-90:0.3) ;
\end{tikzpicture}
}
}
\caption{A counterexample for Center-of-Mass algorithms.}
\label{fig:counterexample}
\end{wrapfigure}

We will actually prove the following claim; from this claim, the theorem holds
since the oracles give identical responses unless $\epsilon=0$.
\begin{claim}
  For any $\epsilon > 0$, there exist functions $f,h$ and a rotation-invariant
  distribution $\mu$ such that $h$ is a balanced halfspace, $f$ is $1/4$-far in
  $\mu$ from all halfspaces, and there exist center-norm oracles $C_{f,\mu},
  C_{h,\mu}$ which on query $\epsilon$ produce random variables $C_f, C_h$ with
  the same distribution.
\end{claim}
\begin{proof}
Let $\sigma_{\sqrt n}$ be the uniform distribution over the sphere of radius
$\sqrt n$ and $\sigma_r$ be the uniform distribution over the sphere of radius
$r$, to be chosen later. Let $\mu = \frac{1}{2}(\sigma_{\sqrt n} + \sigma_r)$.
This distribution is clearly rotation-invariant, and we can transform it into
an isotropic distribution by scaling it appropriately, by a factor of at most
2.  Define $h(x) = \sign(x_1)$, and define $f$ as follows:
\[
  f(x) = \begin{cases}
    h(x) &\text{ if } \norm{x} = \sqrt n \\
    -h(x) &\text{ if } \norm{x} = r
  \end{cases} \,.
\]
We will first show that $f$ is $1/4$-far from all halfspaces.  Let
$g(x) = \sign(\inn{w,x}-t)$ for some arbitrary $w,t$; without loss of
generality, assume $t \leq 0$. Now choose any unit vector $u$ satisfying
$\inn{w,u} \geq 0$ and any $a,b$ satisfying $0 \leq a \leq r < b$; we have $g(
a \cdot u ) = g( b \cdot u ) = 1$ since $\inn{w,u} \geq 0 \geq t$. But $f( r
\cdot u ) \neq f( \sqrt{n} \cdot u)$ so $f$ and $g$ disagree on at least one
of those points; $\Pruc{}{f(x) \neq g(x)}{\inn{x,u} = \norm{x}} \geq 1/2$.
Since half of all unit vectors $u$ satisfy $\inn{w,u} \geq 0$, we have
$\Pr{f(x) \neq g(x)} \geq 1/4$, proving that $f$ is $1/4$-far from
any halfspace.

We also have $\abs{\norm{\Ex{xh(x)}}_2 - \norm{\Ex{xf(x)}}_2} \leq
\norm{\Ex{xf(x)} - \Ex{xh(x)}}_2 \leq r$ since $f,h$ differ only on the points
with $\norm{x}_2 = r$. Thus, setting $r < \epsilon/2$ we can define the
center-norm oracles by defining $C_{f,\mu}(\epsilon) = C_{h,\mu}(\epsilon)$ to
be the random variable uniformly distributed over $\norm{\Ex{xh(x)}}_2 \pm
\epsilon$.
\end{proof}
\ignore{
The Statistical Query model of learning (\cite{Kearns98}) is a model where the
algorithm can make requests for estimates of expectations of functions. Kearns
shows that there is a Statistical Query algorithm for learning balanced
halfspaces over rotation-invariant distributions using $n$ statistical queries
(since statistical queries can be used to estimate the center of mass
coordinates $\Ex{x_i f(x)}$). It is not clear if a statistical query algorithm
can efficiently test halfspaces (even in the balanced case).
}

To get good bounds on the center-norms of functions on RI spaces, we will
quantify the ``maximum concentration'' of distributions with a quantity we call
the \emph{width}; intuitively, the width at $\epsilon$ of a distribution over
$\bR$ is the size of the smallest interval with measure at least $\epsilon$.

\begin{definition}[Width]
\label{def:width}
We will use the L\'evy anticoncentration function (see e.g.~\cite{DS13}) which
is defined as follows: let $\mu$ be an arbitrary distribution over $\bR^n$, let
$w \in \bR^n$ and $r \in \bR, r > 0$. Then
\[
  p_r(w) \define \sup_{\theta \in \bR} \Pru{x \sim \mu}{\abs{\inn{w,x}-\theta}
  \leq r} \,.
\]
Using this function we define the width: for any $\epsilon \in (0,1)$,
\[
W_\mu(w, \epsilon) \define \inf\{r > 0 : p_r(w) \geq \epsilon \} \,.
\]
\end{definition}
\begin{example}
Consider the Gaussian distribution over $\bR^n$: we may ignore $w$ since the
distribution is rotation-invariant, and we can see that for any $r > 0$, the
maximum in $p_r(w)$ is achieved at $\theta
= 0$. Then for any $r$ such that $p_r(w) \geq \epsilon$ we have
$\epsilon \leq \Pru{x \sim \cN(0,1)}{\abs{x} \leq r} \leq \frac{1}{\sqrt{2\pi}} r$
since $1/\sqrt{2\pi}$ is the maximum density of $\cN(0,1)$.  Thus the Gaussian
distribution has $W(w,\epsilon) \geq C\epsilon$ for some constant $C > 0$.
\end{example}

\subsection{Bounds on the Center-Norm}
Equipped with our definition of width, we will now state the main theorem that
allows the center-of-mass algorithm to work. This theorem relates the distance
of a function from a halfspace to the difference between the two centers of
mass. Similar theorems have been proven for the Gaussian space in earlier
papers: see for example Theorem 29 in \cite{MORS10} and the proof of Corollary 4
in \cite{Eld15}. There have also been similar observations about the uniform
distribution over the hypercube \cite{DDFS14, Gol06, OS11}. Our theorem has the
advantage of making no assumption on the distribution; applying the theorem with
the bound $W(w,\epsilon) \geq C\epsilon$ for the Gaussian distribution, as
discussed above, will reproduce the theorems of \cite{Eld15, MORS10} up to
constant factors. The proof is also very simple.

\textbf{Theorem \ref{thm:gap theorem}.}\textit{
Let $\mu$ be any distribution over $\bR^n$ and let $f : \bR^n \to \pmset$ be any
measurable function. Suppose $h(x) = \sign(\inn{w,x}-t)$ is a halfspace
such that $\Ex h = \Ex f$, and let $\epsilon = \dist_\mu(f,h)$. Write $\alpha$
for the angle between $\Ex{x h(x)} - \Ex{x f(x)}$ and $w$. Then
\[
  \norm{\Ex{x h(x)} - \Ex{x f(x)}}_2 \geq \frac{\epsilon}{\cos\alpha}
  W_\mu(w, \epsilon/2) \,.
\] }%
\begin{proof}
Let $\mu_\pi$ be the distribution over $\bR^n$ obtained by projecting $\mu$
onto the vector $w$ (recall that, at the moment, $\mu$ is not necessarily RI).
We rewrite the left side of the theorem as an inner product:
\[
  \norm{\Ex{x(h(x)-f(x))}}_2 \cos \alpha = \abs{\inn{w, \Ex{x(h(x)-f(x))}}} \,.
\]
We now give a lower bound on this inner product.  Write
\[
A^+ \define \{ x : h(x)=1, f(x)=-1 \} \qquad,\qquad
A^- \define \{ x : h(x) = -1, f(x)=1 \}
\]
and note that $\forall x \in A^+, \inn{w,x} \geq t$ and $\forall x \in A^-,
\inn{w,x} < t$. Note that $\mu(A^+) = \mu(A^-) = \epsilon/2$ since $\Ex h = \Ex
f$. Let $X,Y$ be the following conditional random variables, where $x,y \sim
\mu$ are independent:
\[
X \define \inn{w,x} \mid x \in A^+  \qquad,\qquad
Y \define \inn{w,x} \mid x \in A^-  \,.
\]
Then
\begin{align*}
\inn{w, \Ex{x(h(x) - f(x))}}
= \Ex{(h(x)-f(x)) \inn{w,x}}
= 2\left(\mu(A^+) \Ex X - \mu(A^-) \Ex Y\right) \,.
\end{align*}
Let $m_X, m_Y$ be the medians of $X,Y$ respectively. Since $X$ is supported on
values at least $t$ and $Y$ is supported on values at most $t$, we have
\[
  \Ex{X} \geq \frac{1}{2}m_X + \frac{1}{2} t \qquad,\qquad
  \Ex{Y} \leq \frac{1}{2}m_Y + \frac{1}{2} t \,.
\]
Then since $\mu(A^+)=\mu(A^-)=\epsilon/2$ we get
\begin{align*}
  2\left(\mu(A^+)\Ex{X} - \mu(A^-)\Ex Y \right)
  \geq \left( \mu(A^+)m_X  - \mu(A^-)m_Y + (\mu(A^+)-\mu(A^-))t \right) 
  = \frac{\epsilon}{2}(m_X - m_Y) \,.
\end{align*}
Now $\mu_\pi(A^+ \cap [t,m_X]), \mu_\pi(A^- \cap [m_Y,t]) \geq \epsilon/4$ since
$m_X,m_Y$ are the medians of $X,Y$, so
$\mu_\pi[t, m_X] \geq \epsilon/4, \mu_\pi[m_Y, t] \geq \epsilon/4$ and
$\mu_\pi[m_Y, m_X] \geq \epsilon/2$. Selecting $r = (m_X-m_Y)/2$ we see that
$p_r(w) \geq \mu_\pi[m_Y, m_X] \geq \epsilon/2$ so $m_X-m_Y = 2r \geq
2W_\mu(w,\epsilon/2)$, which gives us the lower bound of
$\epsilon \cdot W_\mu(w, \epsilon/2)$.
\end{proof}
\begin{remark}
Using the fact that the density of an isotropic log-concave distribution is
bounded by 1 \cite{LV03} we also get $W_\mu(w,\epsilon) \geq \epsilon$ for any
unit vector $w$ and isotropic log-concave distribution $\mu$.
\end{remark}
\begin{remark}
While the above theorem holds for discrete distributions, it may not be useful:
e.g.~for the uniform distribution over $\pmset^n$ and $\epsilon = 2^{-n}$ we
have $p_0(w) \geq \epsilon$ for every $w$, so $W(w,\epsilon)=0$.
\end{remark}
We define the class of \emph{bounded} RI distributions for which the
center-of-mass tester will work. The general algorithm will partition an
arbitrary RI space into a number of these bounded distributions, allowing us to
use the center-of-mass tester.
\begin{definition}[Bounded RI Distributions]
\label{def:bounded distributions}
  An RI distribution $\mu$ over $\bR^n$ is a \emph{bounded RI distribution} if
  for some radius $R > 0$ and constant $0 < C \leq 1$,
  $\Pru{x \sim \mu}{C R \leq \norm{x} \leq R} = 1$.
\end{definition}
We can easily show that bounded RI distributions have width at least linear in
$\epsilon$; this will be important for the \textsc{Simple-Tester} in the next
section. 
\begin{proposition}
\label{prop:simple ri has linear width}
Let $\mu$ be an isotropic bounded RI distribution with parameters $R,C$. Then
for all $\epsilon > 0$ and any $w, W_\mu(w,\epsilon) \geq C \epsilon$.
\end{proposition}
\begin{proof}
Since $\Ex{\norm{x}^2} = n$ (Proposition \ref{prop:isotropy}) we must have $R
\geq \sqrt n$. By rotation invariance we can drop the direction $w$ from the
definition of width and consider the 1-dimensional projection $\mu_\pi$. Since
$\Pr{\norm x < CR} = 0$ we can bound the maximum density of $\mu_\pi$ by the
maximum density of $(\sigma_{CR})_\pi$ where $\sigma_{CR}$ is the uniform
distribution over the sphere of radius $CR$. By Proposition \ref{prop:density of
sphere} this density is at most $\frac{\sqrt n}{\sqrt{2\pi} CR} \leq \frac{1}{C
\sqrt{2\pi}}$. Then for any interval $[t-r, t+r]$ the total probability mass is
at most $2r \frac{1}{C \sqrt{2\pi}}$; if $r < C \frac{\sqrt{\pi}}{\sqrt 2}
\epsilon$ then this mass is less than $\epsilon$, so we must have
$W_\mu(\epsilon) \geq \epsilon \cdot C \sqrt{\pi/2}$.
\end{proof}
Theorem \ref{thm:gap theorem} proves that a small gap in center-norms implies
close to proximity to a halfspace. This is sufficient for a testing algorithm,
but for the general algorithm we will also require a small amount of tolerance
in our tester; that is, we will need to prove that the algorithm accepts not
only halfspaces but functions that are very close to being halfspaces. For this
purpose we will need an upper bound on the center-norm gap, which is provided by
the next lemma.
\begin{lemma}
\label{lemma:upper bound on norm distance}
Let $\mu$ be any RI distribution with $\Pr{\norm{x} \leq R}=1$ and let $f :
\bR^n \to \pmset$ be any measurable function. Suppose $h$ is a halfspace such
that $\dist_\mu(f,h) = \epsilon$. Then
\[
\norm{\Ex{x h(x)}} - \norm{\Ex{x f(x)}} \leq \BigO{\frac{R}{\sqrt{n}} \epsilon
  \sqrt{\ln(1/\epsilon)}} \,.
\]
\end{lemma}
\begin{proof}
First assume that $\mu$ is the uniform distribution over the sphere of radius
$\sqrt{n}$.  Let $w$ be the normal vector of the halfspace and let $\mu_w$ be
the 1-dimensional projection of $\mu$ onto $w$. Let $A = \{ x : h(x) = 1, f(x) =
-1 \}$ and $B = \{ x : h(x) = -1, f(x) = 1 \}$ so $\epsilon = \mu(A \cup B)$.
Clearly we have $\norm{\Ex{xh(x)}} = \inn{w, \Ex{xh(x)}}$, while
$\norm{\Ex{xf(x)}} \geq \inn{w, \Ex{xf(x)}}$, so we can get an upper bound on
the difference as follows:
\begin{align*}
  \norm{\Ex{xh(x)}} - \norm{xf(x)}
  &\leq \inn{w, \Ex{xh(x)} - \Ex{x f(x)}} \\
  &= \left\langle w, \mu(A)\Exuc{}{x}{x \in A} - \mu(B)\Exuc{}{x}{x \in B} \right.\\
    &\qquad\left. - \mu(B)\Exuc{}{x}{x \in B} + \mu(A)\Exuc{}{x}{x \in A}
      \right\rangle\\
  &= \epsilon \inn{w, \Exuc{}{x}{x \in A} - \Exuc{}{x}{x \in B}} \\
  &\leq 2\epsilon \max_{S : \mu(S) = \epsilon} \Exuc{}{\inn{w,x}}{x \in S} \,.
\end{align*}
This maximum is achieved when $S = \{ x : \inn{w,x} > t_0 \}$ for $t_0$ chosen
such that $\Pr{\inn{w,x} > t_0} = \epsilon$. Then we can get a bound on the
expectation as follows.  Let $t_k$ be the threshold such that $\Pr{\inn{w,x} >
t_k} = \epsilon/2^k$.  Then, using Proposition \ref{prop:tail bounds for
spheres}, $\epsilon/2^k =
\Pr{\inn{w,x} > t_k} \leq \sqrt{2}e^{-t_k^2/4}$ so $t_k \leq \sqrt{4\ln(
\sqrt{2} \cdot 2^k/\epsilon)}$. Thus
\begin{align*}
 \Exuc{}{\inn{w,x}}{\inn{w,x} > t} 
 &= \sum_{k = 1}^\infty \Pruc{}{t_{k-1} < \inn{w,x} \leq t_k}{t < \inn{w,x}}
 \Exuc{}{\inn{w,x}}{t_{k-1} < \inn{w,x} \leq t_k} \\
 &\leq \sum_{k = 1}^\infty \frac{1}{2^k} t_k
 \leq \sum_{k \geq 1} \frac{1}{2^k} \sqrt{2(k + \ln(\sqrt{2} /\epsilon))}
 \leq
  \left( 2 \sum_{k=1}^\infty \frac{k+\ln(\sqrt{2}\epsilon)}{2^k} \right)^{1/2} \\
 &= \sqrt{2(2+\ln(\sqrt{2}/\epsilon))} \,.
\end{align*}
The last inequality is Jensen's inequality, and the final equality is due to
the identity $\sum_{k=1}^\infty k2^{-k} = 2$.  Thus we have an upper bound of $
\epsilon \cdot O(\sqrt{\ln(1/\epsilon)}) = 
\BigO{\epsilon \sqrt{\ln(1/\epsilon)}}$.

Now suppose $\mu$ is any bounded RI distribution. Then the largest sphere
supported by the distribution is of radius $R$, so we can simply multiply the
bound by $R/\sqrt{n}$, since the values on the left scale linearly with the
radius.
\end{proof}

To conclude this section, we give two propositions that bound the center-norm.
\begin{proposition}
\label{prop:center is small}
  Let $\mu$ be any isotropic RI distribution over $\bR^n$ and let $f : \bR^n \to
  \pmset$ be any measurable function. Then $\norm{\Ex{x f(x)}}_2 \leq 1$.
\end{proposition}
\begin{proof}
  Let $\mu_\circ$ be the radial distribution for $\mu$. Suppose $u$ is the
  unit vector parallel to $\Ex{x f(x)}$. Then using Jensen's inequality we have
  \[
    \norm{\Ex{x f(x)}} = \inn{u, \Ex{x f(x)}}
    = \Ex{\inn{u,x}f(x)} \leq \Ex{\abs{\inn{u,x}}} \leq \sqrt{\Ex{\inn{u,x}^2}}
    =1 \,. \qedhere
  \]
\end{proof}

\begin{proposition}
\label{prop:simple ri halfspaces have large centers}
Let $\mu$ be any RI distribution such that for some constant $C > 0$ and for all
$\epsilon > 0, W_\mu(\epsilon) \geq C\epsilon$. Then for any halfspace $h$ with
$\abs{\Ex h} \leq 1-\epsilon$ we have $\norm{\Ex{xh(x)}}_2 \geq C\epsilon/4$.
\end{proposition}
\begin{proof}
By rotation invariance we may assume that the normal vector of $h$ is $e_1$.
Then
\begin{align*}
\norm{\Ex{xh(x)}}
&= \norm{ \Pr{x_1 \geq t }\Exuc{}{x}{x_1 \geq t}
      - \Pr{x_1 < t}\Exuc{}{x}{x_1 < t}} \\
      &= 2\Pr{x_1 \geq t} \norm{\Exuc{}{x}{ x_1 \geq t}}
      \geq \epsilon \Ex{|x_1|}
\end{align*}
since $\Pr{x_1 \geq t} \geq \epsilon/2$.
By taking the median, we have
$\Ex{|x_1|} \geq \frac{1}{2} W_\mu(1/2) \geq C/4$.
\end{proof}

\section{A Tester for Simple Distributions}
\label{section:simple tester}
We now implement \textsc{Simple-Tester}, a center-of-mass tester for bounded RI
distributions.  The main algorithmic ingredient is the subroutine
\textsc{Estimate-IP} in the next subsection, that can be used for estimating the
center-norm of a function and also, in the \textsc{Estimate-Halfspace-Norm}
subroutine, to estimate the center-norm of the nearest halfspace.


\subsection{Estimating Inner Products}
\label{subsection:estimate ip}

\textsc{Estimate-IP} estimates the quantity $\Ex{f(x)g(y)\inn{x,y}}$ for two
independently random vectors $x,y$ (whose distributions may be different). We
could, of course, use standard empirical estimation to estimate this quantity by
picking $m$ pairs $(x,y)$ and computing $f(x)g(y)\inn{x,y}$ for each of the $m$
pairs, but with $m$ samples points from each distribution we actually have $m^2$
pairings available; exploiting this fact lets us achieve $\sqrt n$ sample
complexity.
\begin{algorithm}[H]
  \label{alg:estimate center}
  \caption{\textsc{Estimate-IP}($\mu_1, \mu_2, f, g, m$)}
  \begin{algorithmic}[1]
    \State Draw $\{x_1, \dotsc, x_m\} \sim \mu_1^m$
    \State Draw $\{y_1, \dotsc, y_m\} \sim \mu_2^m$
    \State \Return $\tilde p \gets m^{-2} \sum_{i,j} f(x_i)g(y_j)\inn{x_i,y_j}$
  \end{algorithmic}
\end{algorithm}
\begin{lemma}
\label{lemma:estimate ip}
  Let $\mu_1, \mu_2$ be any RI distributions over $\bR^n$ with $\tau_1 = \bE_{x \sim
  \mu_1}[x_1^2], \tau_2 = \bE_{y \sim \mu_2}[y_1^2]$ satisfying $\tau_1\tau_2
  \leq 1$, let $f,g :
  \bR^n \to \pmset$ be any measurable functions, and let
  $\epsilon,\delta > 0$. Write $p = \inn{\bE_{x \sim \mu_1, f}[x f(x)], \bE_{y
  \sim \mu_2, g}[y g(y)]}$. Then for some universal constant $L$, and arbitrary
  $0 < \epsilon,\delta < 1$,
  \begin{enumerate}
    \item If $p \geq \eta$ and
      $m \geq L \frac{\sqrt n}{\epsilon^2 \eta^2} \log(1/\delta)$ then with
      probability at least $1-\delta$,
      $(1-\epsilon)p \leq \tilde p \leq (1+\epsilon)p$; and,
    \item If $m \geq L \frac{\sqrt n}{\epsilon^2}\log(1/\delta)$ then with
      probability at least $1-\delta$, $\abs{p - \tilde p} \leq \epsilon$.
  \end{enumerate}
\end{lemma}
\begin{proof}
  By rotation invariance we have $\bE_{x \sim \mu_i}[\inn{x,u}^2] = \tau_i$ for
  each $i \in \{1,2\}$ and any unit vector $u$.

  We will use Chebyshev's inequality. Let $\{x_i\}_{i \in [m]}, \{y_j\}_{j \in
  [m]}$ be the sets of random points that the algorithm receives. We will write
  $X_{i,j} \define f(x_i)g(y_j)\inn{x_i, y_j}$. Then
  \[
    \Ex{X_{i,j}} = \Ex{\inn{x_i f(x_i), y_j g(y_j)}} = \inn{\Ex{xf(x)},\Ex{yg(y)}}
    = p
  \]
  so, since $\tilde p = m^{-2} \sum_{i , j} X_{i,j}$,
  we have $\Ex{\tilde p} = p$.  By Chebyshev's Inequality, we get a
  bound for each the two desired conclusions; for the multiplicative error:
  \begin{equation}
  \label{eq:estimate ip multiplicative error}
    \Pr{\tilde p > (1+\epsilon)\Ex{\tilde p} \text{ or } \tilde p <
    (1-\epsilon)\Ex{\tilde p}}
    = \Pr{\abs{\tilde p - \Ex{\tilde p}} > \epsilon \Ex{\tilde p}}
    \leq \frac{\Var{\tilde p}}{\epsilon^2 \Ex{\tilde p}^2}
    \leq \frac{\Var{\tilde p}}{\epsilon^2 \eta^2}
  \end{equation}
  and for the additive error
  \begin{equation}
  \label{eq:estimate ip additive error}
    \Pr{\abs{\tilde p - \Ex{\tilde p}} > \epsilon} \leq \frac{\Var{\tilde p}}{\epsilon^2} \,.
  \end{equation}
  We will compute the variance:
  \begin{align*}
    \Var{\tilde p}
    &= m^{-4} \sum_{i,j, k,\ell} \Cov(X_{i,j}, X_{k,\ell}) \,.
  \end{align*}
  If $i \neq k$ and $j \neq \ell$, the covariance is 0 since $X_{i,j},
  X_{k,\ell}$ are independent. When either $i = k$ or $j = \ell$
  (say $j = \ell$), we have
  \begin{align*}
  \Cov(X_{i,j},X_{k,\ell})
  = \Ex{X_{i,j}X_{k,j}} - \Ex{X_{i,j}}^2
  &\leq \Ex{f(x_i)f(x_k)g(y_j)^2 \inn{x_i, y_j}\inn{x_k, y_j}} \\
  &= \Ex{\inn{y_j, \Ex{xf(x)}}^2} \\
  &= \norm{\Ex{xf(x)}}^2 \Ex{\inn{y_j, \frac{\Ex{xf(x)}}{\norm{\Ex{xf(x)}}}}^2} \\
  &= \tau_2 \norm{\Ex{x f(x) }}^2 \leq \tau_2 \tau_1 = 1
  \end{align*}
  where the bound is due to Proposition \ref{prop:center is small}; in the case
  of $i = k$ we get the same bound.  This situation occurs $m^3$ times. Finally,
  for $i=k,j=\ell$, we may use the identity $\bE[\norm{y_j}^2] = n\tau_2$
  (Proposition \ref{prop:isotropy}) to get
  \begin{align*}
    \Cov(X_{i,j}, X_{k,\ell})
    &= \Var{X_{i,j}} 
    = \Ex{f(x_i)^2 g(y_j)^2 \inn{x_i, y_j}^2} - p^4 
    = \Ex{ \norm{y_j}^2 \inn{x_i, y_j/\norm{y_j}}^2} - p^4 \\
    &= n \tau_1 \tau_2 - p^4 \leq n \,,
  \end{align*}
  This situation occurs $m^2$ times, so in total the variance is at most
  \[
  \Var{\tilde p} \leq m^{-4} \left( m^3 + m^2 n \right)
  = \frac{1}{m} + \frac{n}{m^2} 
  \leq 2\frac{n}{m^2} \xh{(for $m < n$)} \,.
  \]
  Thus setting $m = O\fracb{\sqrt n}{\epsilon^2\eta^2}$ for the multiplicative
  error and $m = O\fracb{\sqrt n}{\epsilon^2}$ for the additive error suffices
  to a bound of $1/3$ on inequalities \eqref{eq:estimate ip multiplicative
  error} and \eqref{eq:estimate ip additive error}.

  Finally, we may apply a classic boosting technique and repeat this process $M$
  times, taking the median. The probability of failure is bounded by the
  probability that at least half of the trials fail; by the Chernoff bound this
  is at most $\BigO{\exp{-M}}$ so it succeeds with probability at least
  $1-\delta$ when $M = \BigO{\log(1/\delta)}$.
\end{proof}
To simplify the presentation of later algorithms, we define the following
wrapper around \textsc{Estimate-IP} which estimates the center-norm of a
function $f$:
\begin{algorithm}[H]
  \label{alg:estimate norm}
  \caption{\textsc{Estimate-Norm}($\mu, f, m$)}
  \begin{algorithmic}[1]
    \State \Return $\textsc{Estimate-IP}(\mu, \mu, f, f, m)$
  \end{algorithmic}
\end{algorithm}

\subsection{The \textsc{Simple-Tester} Algorithm}

The previous section shows how to estimate the center-norm of an arbitrary
function. We must compare that estimate to an estimate of the center-norm of a
halfspace with the same mean; recall that due to rotation invariance we can
ignore the orientation of the halfspaces since it does not affect their
center-norms. We will now show how to estimate the center-norm of a halfspace $h$
with mean $v$, being given only an estimate $\tilde v$ of $v$. 

We first show that the center-norm of a halfspace with mean $v$ is not very
sensitive to changes in $v$.  We achieve this by taking the derivative of the
center-norm with respect to the mean of the halfspace.
\begin{definition}
  Let $\mu$ any RI distribution over $\bR^n$, and write $\Phi$ for the
  ``two-sided CDF'' of the 1-dimensional projection $\mu_\pi$:
  \[
  \Phi(t) \define \Pru{z \sim \mu_\pi}{z \geq t} - \Pru{z \sim \mu_\pi}{z < t}
    = 2\Pru{z \sim \mu_\pi}{z \geq t} - 1
    = 2 \int_t^\infty d\mu_\pi(z) - 1 \,.
  \]
  Then for all $v \in [-1,1]$ we define $\xi$ as the center-norm of the
  halfspace with mean $v$ (recall that by rotation invariance, the orientation
  of the halfspace does not affect the center-norm):
  \[
  \xi(v) \define \Exu{z \sim \mu_\pi}{ z \sign(z - \Phi^{-1}(v)) }
  = 2\int_{\Phi^{-1}(v)}^\infty z d\mu_\pi(z) \,.
  \]
\end{definition}

\begin{proposition}
\label{prop:derivative of xi}
Let $\mu$ be an isotropic RI distribution. Then
  $\abs{\frac{d}{dx}\xi(x)^2} \leq 1 $.
\end{proposition}
\begin{proof}
  Write $\theta = \Phi^{-1}(x)$. Observe that $\frac{d\theta}{dx} =
  \fracb{dx}{d\theta}^{-1} = \left(-2 d\mu_\pi(\theta)\right)^{-1}$. The
  derivative of $\xi(x)$ is:
  \[
  \frac{d}{dx}\xi(x) = -2\theta d\mu_\pi(\theta) \frac{d\theta}{dx}
  = \theta \,.
  \]
  From here we have
  \[
    \frac{d}{dx}(\xi(x))^2 = 2\xi(x) \frac{d}{dx}\xi(x)
    = 2\theta\xi(x)
    = 2\theta\int_\theta^\infty zd\mu_\pi(z)
    \leq 2 \int_\theta^\infty z^2d\mu_\pi(z)
    \leq \Exu{z \sim \mu_\pi}{z^2} = 1 \,.
  \]
  where the last equality is by isotropy.
\end{proof}

Using this fact, we show that we can estimate $\norm{\Ex{ xh(x)}}_2^2$ given an
estimate of $\Ex h$.
\begin{algorithm}[H]
\label{alg:estimate halfspace norm}
\begin{algorithmic}[1]
\State Draw $X \sim \mu^m$ for $m = \frac{1}{2(\epsilon/2)^2}\ln(4/\delta)$;
\State $q \gets
  L\frac{\sqrt n}{(\epsilon/2)^2}\log(2/\delta)$ for $L$ in Lemma
  \ref{lemma:estimate ip}.
\State $\tilde t \gets
  \max\{ t : \frac{1}{m} \#\{ x \in X : x_1 \geq t \} \geq (1-\tilde v)/2 \} $
\State $\tilde h(x) \define \sign(x_1 - \tilde t)$
\State \Return $\tilde p \gets \textsc{Estimate-Norm}(\mu_\pi, h, q)$
\end{algorithmic}
\caption{\textsc{Estimate-Halfspace-Norm}($\mu, \tilde v, \epsilon, \delta$)}
\end{algorithm}
\begin{lemma}
\label{lemma:estimate halfspace norm}
Let $\epsilon, \delta > 0$, let $\mu$ be any isotropic RI distribution over
$\bR^n$ and let $h$ be a halfspace with $\Ex h = v$.  Suppose $\abs{\tilde v -
v} < \epsilon$. Then with probability at least $1-\delta$,
\textsc{Estimate-Halfspace-Norm} produces an estimate $\tilde p^2$ satisfying
$\abs{\tilde p^2 - \norm{\Ex{x h(x)}}_2^2} < \epsilon$, and uses at most
$\BigO{\frac{1}{\epsilon^2}\log(1/\delta)}$ random samples.
\end{lemma}
\begin{proof}
The bound on the number of samples is achieved by definition of $m$ and $q$. By
Lemma \ref{lemma:estimate ip}, \textsc{Estimate-Norm} fails with probability
at most $\delta/2$.

Let $t$ be the threshold of $h$, and let $t'$ be the threshold such that the
halfspace $h'$ with threshold $t'$ satisfies $\Ex{h'} = \tilde v$. By the
guarantee on $\tilde v$ we have $\hdist(t, t') \leq \frac{1}{2}\abs{v - \tilde
v} \leq \epsilon/2$. From Proposition \ref{prop:estimate threshold} we know that
$\hdist(t', \tilde t) < \epsilon/2$ so $\hdist(\tilde t, t) \leq \hdist(\tilde t,
t') + \hdist(t', t) \leq \epsilon$.

From the bound on the derivative of $\xi$ (Proposition \ref{prop:derivative of xi})
we see that $\norm{\bE[x \tilde h(x)]}^2$ has error at most
\[
\abs{\bE[\tilde h] - \Ex h}
  \leq \frac{1}{2}\hdist(\tilde t, t)
  \leq \epsilon/2 \,,
\]
and from \textsc{Estimate-IP} (Lemma \ref{lemma:estimate ip}) we know that
$\tilde p$ is within $\pm \epsilon/2$ of $\norm{\bE[x \tilde h(x)]}^2$. Thus we
conclude that $\tilde p$ is within $\epsilon$ of $\norm{\Ex{x h(x)}}^2$.
\end{proof}

Finally we compose these estimations to get a tester for bounded RI
distributions. In the general algorithm, we will require a small amount of
tolerance in our tester; i.e.~the tester must accept halfspaces and also any
function that is very close to being a halfspace.
\begin{algorithm}[H]
\label{alg:large width tester}
\caption{\textsc{Simple-Tester}$(\mu, f, \epsilon, \delta)$}
\begin{algorithmic}[1]
\State $\epsilon_1^3 \gets K_1 C^2 \epsilon^3$;
  $\epsilon_3^3 \gets K_3 C^2 \epsilon^3$ (for some constants $K_1, K_3$)
\State $\tilde v \gets \textsc{Estimate-Mean}(\mu, f, \epsilon_1^3, \delta/3)$

\State $\tilde c^2 \gets
  \textsc{Estimate-Norm}(\mu, f, m \define L \frac{\sqrt n}{\epsilon_1^6}\log(3/\delta))$
  for $L$ in Lemma \ref{lemma:estimate ip}

\State $\tilde p^2 \gets
  \textsc{Estimate-Halfspace-Norm}(\mu, \tilde v, \epsilon_1^3, \delta/3)$

\If{$\tilde p^2 - \tilde c^2 < \epsilon_3^3$ or $|\tilde v| \geq 1-\epsilon$}
  \textbf{accept}
\EndIf
\end{algorithmic}
\end{algorithm}
\begin{theorem}
\label{thm:simple tester}
Let $C$ be a constant independent of $n$ and $R = \sqrt{n}$. For any $\eta > 0$
there exists a constant $K_2(\eta)$ depending only on $\eta$ so that
\textsc{Simple-Tester} satisfies the following properties: Let $\mu$ be any
bounded RI distribution over $\bR^n$ with $R=\sqrt n$ and constant $C$, and let
$f : \bR^n \to \pmset$ be a measurable function. Suppose $\epsilon,\delta \in
(0,1/2)$. Then for $\epsilon_2 = K_2(\eta) \epsilon^{3+\eta}$,
\begin{enumerate}
  \item If $f$ is $\epsilon_2$-close to a halfspace, the the algorithm accepts
    with probability at least $1-\delta$;
  \item With probability at least $1-\delta$, if the algorithm accepts $f$ then
    there exists a halfspace $h$ aligned with $f$ satisfying $\Ex{h}=\Ex{f}$ and
    $\dist(f,h) \leq \epsilon$.
\end{enumerate}
Furthermore, the algorithm requires at most $\BigO{\frac{\sqrt
n}{\epsilon^6}\log(1/\delta)}$ random samples.
\end{theorem}
\begin{proof}
  Assume without loss of generality that $\mu$ is istropic, since we may scale
  the distribution.

  \textsc{Estimate-Mean} requires at most
  $\BigO{\frac{1}{\epsilon^6}\log(1/\delta)}$ samples and fails with probability
  at most $\delta/3$ (Lemma \ref{lemma:estimate
  mean}), \textsc{Estimate-Halfspace-Norm} requires at most $\BigO{\frac{\sqrt
  n}{\epsilon_1^6}\log(1/\delta)}$ and fails with probability at most $\delta/3$
  (Lemma \ref{lemma:estimate halfspace norm}), and \textsc{Estimate-Norm}
  uses $m = \BigO{\frac{\sqrt n}{\epsilon_1^6}\log(1/\delta)}$ samples and fails
  with probability at most $\delta/3$ (Lemma \ref{lemma:estimate ip}); assume
  these estimations all succeed, which occurs with probability at least
  $1-\delta$.

  \textbf{Completeness:} Suppose $f$ is $\epsilon_2$-close to a halfspace $h$.
  First suppose $\abs{\Ex f} \geq 1-\epsilon/2$. Then by the guarantee on
  \textsc{Estimate-Mean} (Lemma \ref{lemma:estimate mean}) we have $|\tilde v|
  \geq 1-\epsilon$ so the algorithm accepts. Now suppose $\abs{\Ex f} <
  1-\epsilon/2$.

  For convenience, let $v = \bE_{\mu}{f}$ and write $p = \norm{\Ex{xh(x)}}, c =
  \norm{\Ex{x f(x)}}$.

  From Proposition \ref{prop:center is small} we have $p,c \leq 1$. Letting $K$
  be the constant in Lemma \ref{lemma:upper bound on norm distance} and using
  that lemma with the fact that $\dist(f,h) < \epsilon_2$ and $R=\sqrt n$, we
  have
  \[
  (p^2 - c^2) = (p+c)(p-c)
  \leq (p+c) \cdot K \epsilon_2 \sqrt{\ln(1/\epsilon_2)}
  \leq 2K \epsilon_2 \sqrt{\ln(1/\epsilon_2)} \,.
  \]
  Since $n \bE[x_1^2] = \bE[\norm{x}^2] \leq R^2 = n$ we have $\bE[x_1^2] \leq
  1$, so by the guarantee on \textsc{Estimate-Norm} (Lemma \ref{lemma:estimate
  ip}) we have $\tilde p^2 \leq p^2 + \epsilon_1^3$ and $\tilde c^2 \leq c^2 +
  \epsilon_1^3$, and we have $\epsilon_2 = K_2 \epsilon^{3+\eta}, \epsilon_1^3
  = K_1 C^2 \epsilon^3$, for some constants $K_1, K_2$ to be chosen later, so
  \[
    \tilde p^2 - \tilde c^2
    \leq p^2 - c^2 + 2\epsilon_1^3
    \leq 2K \epsilon_2 \sqrt{\ln(1/\epsilon_2)}
      + 2K_1 C^2 \epsilon^3
    = 2K K_2 \epsilon^{3+\eta}
      \sqrt{\ln\fracb{1}{K_2 \epsilon^{3+\eta}}} + 2K_1 C^2 \epsilon^3 \,.
  \]
  We want to show that this is at most $\epsilon_3^3$. Recall $\epsilon_3^3 =
  K_3C^2\epsilon^3$; for the second term we have for $K_1 \leq K_3/4$ that this
  term is at most $\epsilon_3^3/2$, so it suffices to show that the first term
  is also at most $\epsilon_3^3/2$. Then we want to show, for $A = K_3/4K$,
  \[
    K_2 \epsilon^{3+\eta} \sqrt{\ln\fracb{1}{ K_2 \epsilon^{3+\eta}}}
    \leq A\epsilon^3
    \,\equiv\,
    \frac{1}{K_2} \leq \exp{A^2/K_2^2\epsilon^{2\eta} - (3+\eta)\ln(1/\epsilon)}
    \,.
  \]
  For $\eta > 0$ the exponent is bounded so we can choose $K_2 = K_2(\eta)$
  (where $K_2(\eta)$ is a constant depending on $\eta$) so that the above
  inequality holds; from this we conclude that $\tilde p^2 - \tilde c^2 \leq
  K_3\epsilon^3 = \epsilon_3^3$, so the test passes.
  \ignore{
  \[
    \frac{d}{d\epsilon} \epsilon^\frac{6+2\eta}{8} = \frac{3+\eta}{4}
    \epsilon^{2(\eta-1)} = \frac{3+\eta}{4}
  \]
  \[
    \frac{d}{d\epsilon} \exp{-K'/\epsilon^{2\eta}}
    = -K' \exp{-K'/\epsilon^{2\eta}} \frac{d}{d\epsilon} \epsilon^{-2\eta}
    = 2K'\exp{-K'/\epsilon^{2\eta}} \epsilon^{-2\eta-1} \eta
    = 2K' \exp{-K'} \eta
    \geq \frac{3+\eta}{4}
  \]
  when $K' \approx \frac{1}{\eta}\ln(1/\eta)$(??).
  }

  \textbf{Soundness:} First suppose $|\tilde v| \geq 1-\epsilon$. Then $\abs{\Ex
  f} \geq 1-\epsilon-\epsilon^3 \geq 1-2\epsilon$. Let $h$ be any halfspace
  satisfying $\Ex h = \Ex f$; then $\dist(f,h) \leq \epsilon$ (Proposition
  \ref{prop:distance from mean}).

  Now suppose $|\tilde v| < 1-\epsilon$ so $\abs{\Ex f} < 1-\epsilon+\epsilon_1^3
  \leq 1-\epsilon/2$, and let $h$ be the halfspace aligned with $f$ with $\Ex h =
  \Ex f$. By Lemma \ref{lemma:estimate ip} we have
  \[
    \tilde c^2 \in \norm{\Ex{x f(x)}}^2 \pm \epsilon_1^3 = c^2 \pm \epsilon_1^3
  \]
  and from Lemma \ref{lemma:estimate ip},
  \[
    \tilde p^2 \in \norm{\Ex{x h(x)}}^2 \pm \epsilon_1^3 = p^2 \pm \epsilon_1^3\,.
  \]
  Suppose for the sake of contradiction that $\dist(f,h) > \epsilon$. Then by
  Theorem \ref{thm:gap theorem}, Proposition \ref{prop:simple ri has linear
  width}, and the identity $(a+b)(a-b) = a^2 - b^2$, we have
  \[
  \frac{C}{2} \epsilon^2 (p+c)
  \leq (p-c)(p+c)
  \leq \tilde p^2 - \tilde c^2 + 2\epsilon_1^3
  \leq \epsilon_3^3 + 2\epsilon_1^3
  \]
  where the final inequality holds because the test has passed.  From
  Proposition \ref{prop:simple ri halfspaces have large centers} and the fact
  that $\Ex h < 1-\epsilon/2$, we have $p \geq C\epsilon/8$. We also have
  $\epsilon_1^3 = K_1 C^2 \epsilon^3, \epsilon_3^3 = K_3 C^2 \epsilon^3$, so
  \[
    \frac{C^2}{16}\epsilon^3 \leq \epsilon_3^3 + 2\epsilon_1^3
    = (K_3 + 2K_1) C^2 \epsilon^3 < \frac{C^2}{16} \epsilon^3
  \]
  for appropriate choices of constants $K_1, K_3$ (and recall from the
  completeness proof that we require only that $K_1 \leq K_3/4$).  This is a
  contradiction.  Thus $\dist(f,h) \leq \epsilon$.
  \end{proof}
  \begin{remark}
    While we have proven the correctness of this tester for \emph{bounded}
    RI distributions, we could also show that the tester, without the tolerance
    guarantee, works for any RI distribution satisfying $W(w,\epsilon) =
    \Omega(\epsilon)$. An important example would be the isotropic
    \emph{log-concave} RI distributions \cite{LV03}.
  \end{remark}

  It is worth comparing this algorithm and analysis to those provided by Matulef
  \emph{et al.} \cite{MORS10} and Balcan \emph{et al.} \cite{BBBY12}. Those
  algorithms both used the same high-level strategy but were proven to work only
  for the Gaussian distribution. The \cite{MORS10, BBBY12} algorithms had no
  need for the \textsc{Estimate-Halfspace-Norm} subroutine since for the
  Gaussian distribution the center-norm of a halfspace with mean $v$ is given by
  the function $(2\phi(\Phi^{-1}(v)))^2$, where $\phi, \Phi$ are the density and
  CDF of the standard normal distribution \cite{MORS10}; by using a sampling
  algorithm instead, we eliminate the need for such exact relationships.
  Finally, our analysis provides a small tolerance guarantee, which was
  unnecessary for the earlier works.

\section{A Tester for General RI Distributions}
\label{section:ri tester}
We will now show how to use the \textsc{Simple-Tester} as a subroutine to get a
tester for general RI distributions. There are 3 ideas, covered in the next 3
subsections, that we will use to partition the RI space into simple sections:

\begin{enumerate}
\item (Subsection \ref{subsection:find pivot}.) A halfspace with threshold $t$
will be constant on the ball of radius $t$ and almost balanced outside a ball
of radius $T \gg t$. We can quickly identify the values of $t$ and $T$ using the
\textsc{Find-Pivot} algorithm. The middle region can be partitioned into simple
sections while the extreme outer regions will be treated specially.
\item \label{outline projection} (Subsection \ref{subsection:projection}.) The
region outside radius $T \gg t$, where the halfspace is nearly balanced, is
essentially an arbitrary RI distribution. We can normalize all the sample points
in this region so that they all lie in a large bounded space; we show that this
preserves balanced halfspaces. Once they are in a bounded space we can apply
\textsc{Simple-Tester}.

\item (Subsection \ref{subsection:check consistency}.) We can run the
\textsc{Simple-Tester} on each of these bounded regions, but then we must ensure
that the halfspaces in each region are consistent with each other.  We do this
with the \textsc{Check-Consistency} algorithm, which for any two functions that
are close to halfspaces will check that the halfspaces are close to each other.
\end{enumerate}
\ignore{
\begin{wrapfigure}[13]{r}{1.2in}
\centering
\begin{tikzpicture}
  \draw[fill=gray!50] (0,0) circle [x radius=1, y radius=1];
  \draw[fill=white] (0,0) circle [x radius=0.25, y radius=0.25];
  \draw (0,0) circle [x radius=0.4, y radius=0.4];
  \draw (0,0) circle [x radius=0.55, y radius=0.55];
  \draw (0,0) circle [x radius=0.7, y radius=0.7];
  \draw (0,0) circle [x radius=0.85, y radius=0.85];
  \draw (0.25, 1.1) -- (0.25, -1.1) node[anchor=north]{$t$};
\end{tikzpicture}
\caption{
}
\label{fig:outline of tester}
\end{wrapfigure}
}

\subsection{Finding the Important Radii}
\label{subsection:find pivot}
For a function $f$, \textsc{Find-Pivot} identifies a threshold $\tilde t$ (a
``pivot'') such that $f$ is nearly constant on the ball of radius $\tilde t$. It
simply returns the smallest radius that can possibly satisfy this condition
given a set of examples.
\begin{algorithm}[H]
\label{alg:find pivot}
\caption{\textsc{Find-Pivot}$(\mu, f, \epsilon, \delta)$}
\begin{algorithmic}[1]
\State Draw $X \sim \mu^m$ for $m = \frac{1}{\epsilon}\ln(2/\delta)$
\If{$f$ is monochromatic on $X$}
\Return $\infty$
\Else
\State \Return
  $\tilde t
    = \min \{ \norm{x} : x,y \in X, \norm{y} \leq \norm{x}, f(x) \neq f(y) \}$
\EndIf
\end{algorithmic}
\end{algorithm}
\begin{lemma}
\label{lemma:find pivot}
Let $\epsilon, \delta \in (0,1)$.  With probability at least $1-\delta$,
\textsc{Find-Pivot} returns $\tilde t$ satisfying the following:
\begin{enumerate}
  \item For some $b \in \pmset$,
    $\Pr{f(x) \neq b, \norm{x} < \tilde t} < \epsilon$ and $\Pr{f(x)=b, \norm x
    \leq \tilde t} > 0$;
  \item If $f$ is a halfspace with threshold $t$, then $\tilde t \geq |t|$.
\end{enumerate}
\end{lemma}
\begin{proof}
  For all $t \geq 0$, define $b_t \define \sign\left(\Exuc{}{f(x)}{\norm{x} <
  t}\right)$. If $\Pr{f(x) \neq b_t, \norm{x} < t} < \epsilon$ for all $t$ the
  conclusion holds, so assume that $\Pr{f(x) \neq b_t, \norm{x} < t} \geq
  \epsilon$ for some $t$. Then the following minimum exists:
  \[
    T = \min_t \{ \Pr{f(x) \neq b_t, \norm{x} < t} \geq \epsilon \} \,.
  \]
  Clearly we have $\Pr{f(x) = b_t, \norm{x} < t} \geq \Pr{f(x) \neq b_t,
  \norm{x} < t}$ for all $t$ so in particular this holds for $T$.

  If any two points $x,y \in X$ satisfy $f(x) \neq f(y)$ and $\norm{x}, \norm{y}
  < T$ then the algorithm returns $\tilde t$ such that $\tilde t < T$, implying
  $\Pr{f(x) \neq b_{\tilde t}, \norm{x} < \tilde t} < \epsilon$, since otherwise
  there would be a contradiction. So we can bound the failure probability by
  the probability that this event fails. By the union bound, this probability is
  at most
  \begin{align*}
    &\Pr{\forall x \in X : f(x) = b_T \text{ or } \norm{x} \geq T}
    + \Pr{\forall x \in X : f(x) \neq b_T \text{ or } \norm{x} \geq T} \\
    &\qquad \leq 2(1-\epsilon)^m
    \leq 2e^{-\epsilon m}
  \end{align*}
  which is bounded by $\delta$ when $m = \frac{1}{\epsilon}\ln(2/\delta)$.  The
  second conclusion holds because, for a halfspace with a threshold $t$ and
  normal $w$, if two points $x,y$ satisfy $\norm y \leq \norm x, f(x) \neq f(y)$
  then we must have $\norm x \geq \abs{\inn{w,x}} \geq |t|$ so in particular
  $\tilde t \geq |t|$.
\end{proof}

\subsection{Rescaling Arbitrary RI Distributions onto a Bounded Space}
\label{subsection:projection}
Given a set of examples of a balanced halfspace, we can normalize the example
points so that they all lie on the same sphere and preserve the halfspace: a
point $x$ satisfying $\inn{w,x} \geq 0$ also satisfies $\inn{w,x/\norm{x}} \geq
0$; but this is not true for all halfspaces. Here we show that we can perform
this transformation on any halfspace that is close to balanced, with a small
cost to the distance.

Let $\pi_r : (0,\infty) \to (r,2r)$ be the bijection $\pi_r(x) = r(2 - e^{-x})$
and define $\pi : \bR^n \to \bR^n$ as $\pi(x) = x \cdot
\pi_r(\norm{x})/\norm{x}$.

\begin{lemma}
\label{lemma:projecting balanced functions}
Let $\mu$ be a RI distribution over $\bR^n$, let $\epsilon \in (0,1/2)$,
and let $f : \bR^n \to \pmset$ be a measurable function. Then there exists some
radius $r$ such that if $\sigma$ is the distribution of $\pi(x)$ for $x \sim
\mu$, and $g(x) \define f(\pi^{-1}(x))$, then:
\begin{enumerate}
\item If $f$ is a halfspace then $f,g$ are aligned;
\item If $\dist_\mu(f,h) \leq \epsilon$, for some halfspace $h$ with
  $\abs{\Ex h} < \eta$, then $\dist_\sigma(g,h) < \epsilon + \eta$;
\item If $\dist_\sigma(g,h) \leq \epsilon$ for a halfspace $h$
satisfying $\bE_{\sigma}{h}=\bE_{\sigma}{g}$ and $\abs{\bE_{\sigma}{g}} <
\eta$, and $h'$ is the balanced halfspace aligned with $h$, then
$\dist_\mu(f, h') < \epsilon + \eta/2$.
\end{enumerate}
\end{lemma}
\begin{proof}
If $f$ is a halfspace with normal vector $w$ then $f,g$ are aligned since all
points in a set $R_{a,b} = \{ x : \norm{x}=a, \inn{w,x}=b \}$ have the same
function value and are centered on $wb$, which becomes $wbr(2-e^{-a})$ after the
transformation.

Suppose $f$ is a $\epsilon$-close to a halfspace $h$ with $\abs{\Ex h} < \eta$,
say $h(x) = \sign(\inn{e_1,x}-t)$. Assume without loss of generality that $t
\geq 0$. Then
\begin{align*}
\dist_\sigma(g,h)
&= \Pru{x \sim \mu}{g(\pi(x)) \neq h(\pi(x))} \\
&\leq \Pru{x \sim \mu}{h(\pi(x)) = h(x) \wedge f(x) \neq h(x)}
  + \Pru{x \sim \mu}{h(\pi(x)) \neq h(x)} \\
&\leq \epsilon + \Pr{x_1 < t , (\pi(x))_1 \geq t} + \Pr{x_1 \geq t , (\pi(x))_1
< t} \,.
\end{align*}
$\Pr{x_1 < t} = |\bE[h]|/2 \leq \eta/2$ so what remains is to bound the second
probability. We have $(\pi(x))_1 \geq r x_1 / \norm{x}$ so for large enough $r$
this probability will be at most $\eta/2$. In particular if $\Ex{\norm{x}^2} =
n$ then $\Pr{x_1 > 2r} \leq \Pr{\norm x > 2r} \leq n/4r^2$ by Chebyshev's
inequality, so for $r > n$ only an $o(1)$ fraction of points will get closer to
the origin after applying $\pi$.

Now suppose that $h$ is a halfspace satisfying $\dist_\sigma(g,h) \leq \epsilon$
and $\bE_\sigma[h] = \bE_\sigma[g]$, and $|\bE_\sigma[g]| < \eta$. Let $h'$ be
the balanced halfspace aligned with $h$. Then since $h'(x) = h'(\pi(x))$ for all
$x$ we have
\[
\dist_\mu(f,h')
= \dist_\sigma(g,h') \leq \dist_\sigma(g,h) + \dist_\sigma(h,h') 
\leq \epsilon + \eta/2 \,. \qedhere
\]
\end{proof}

\subsection{Checking the Consistency of Two Halfspaces}
\label{subsection:check consistency}
After dividing the space into many simple sections and applying the
\textsc{Simple-Tester}, we need to combine the results to get the final
decision. The main idea is that the distance between two halfspaces can be
decomposed via the triangle inequality into a threshold component and an angle
component (Proposition \ref{prop:distance decomposition}).
\textsc{Check-Threshold} checks that there is a unifying threshold value for two
halfspaces on different domains, and \textsc{Check-Consistency} further enforces
that the angle between two halfspaces is small. Finally, Lemma \ref{lemma:check
consistency group} allows us to apply \textsc{Check-Consistency} pairwise and
get a consistent result for all regions.  The proofs of these lemmas are lengthy
and are left until after the presentation of the complete testing algorithm (see
Subsection \ref{subsection:proofs of consistency checks}).

\begin{algorithm}
  \label{alg:check threshold}
  \caption{\textsc{Check-Threshold}$(\mu_1, \mu_2, v_1, v_2, \epsilon, \delta)$}
  \begin{algorithmic}[1]
    \State Draw $X_1 \sim \mu_1^m, X_2 \sim \mu_2^m$ for
      $m = \BigO{\frac{1}{\epsilon^2}\log(1/\delta)}$
    \For{$i \in \{1,2\}$}
    \State $\tilde a_i \gets
      \max\{ z : \#\{x \in X_i : x \geq z\} \geq m(v_i + 2\epsilon/3) \}$ (if no
      such $z$ exists, $\tilde a_i \gets -\infty$)
    \State $\tilde b_i \gets
      \min\{ z : \#\{x \in X_i : x \geq z\} \leq m(v_i - 2\epsilon/3) \}$
      (if no such $z$ exists, $\tilde b_i \gets \infty$)
    \EndFor
    \State If $[a_1,b_1],[a_2,b_2]$ intersect, \textbf{accept}; else
    \textbf{reject}
  \end{algorithmic}
\end{algorithm}
\begin{lemma}
\label{lemma:check threshold}
Let $\epsilon,\delta \in (0,1/2)$, $v_1, v_2 \in [0,1]$ and let $\mu_1, \mu_2$
be any distributions over $\bR$, unknown to the algorithm. Let $t_1, t_2$
satisfy $\Pru{x \sim \mu_1}{x_1 \geq t_1} = v_1, \Pru{y \sim \mu_2}{y_1 \geq
t_2} = v_2$. Then, using $\BigO{\frac{1}{\epsilon^2}\log(1/\delta)}$ samples,
with probability at least $1-\delta$ \textsc{Check-Threshold} satisfies the
following:
\begin{enumerate}
\item If there exists a threshold $t$ such that $\hdist_{\mu_1}(t_1, t),
\hdist_{\mu_2}(t_2, t) < \epsilon/3$ then \textsc{Check-Threshold} accepts; and
\item If \textsc{Check-Threshold} accepts then there exists $t$ such that
$\hdist_{\mu_1}(t_1, t), \hdist_{\mu_2}(t_2, t) \leq \epsilon$.
\end{enumerate}
\end{lemma}

The idea behind \textsc{Check-Consistency} is simple: we will first ensure that
the ``angle between $f_1, f_2$'' is small by estimating their center-norms and
comparing it with the inner product of their centers of mass; this can be done
with \textsc{Estimate-IP}. This will guarantee that if $f_1,f_2$ are close to
halfspaces $h_1$ and $h_2$, then $h_1$ and $h_2$ have normal vectors pointing in
roughly the same direction. We then make sure that $h_1, h_2$ have consistent
thresholds using \textsc{Check-Threshold}.

\begin{algorithm}[H]
\label{alg:check consistency}
\caption{\textsc{Check-Consistency}$(\mu_1, \mu_2, f_1, f_2, \epsilon, \delta)$}
\begin{algorithmic}[1]
\State $\epsilon_1 \gets \frac{C}{\sqrt{12}} \cdot \epsilon$
\State $\tilde v_1 \gets \textsc{Estimate-Mean}(\mu_1, f_1, \epsilon/3, \delta/6)$
\State $\tilde v_2 \gets \textsc{Estimate-Mean}(\mu_2, f_2, \epsilon/3, \delta/6)$
\State $m \gets \BigTheta{\frac{\sqrt n}{\epsilon_1^4}\log(1/\delta)}$
\If{$|v_1|,|v_2| < 1-\epsilon$}
  \State $\tilde c_1^2 \gets \textsc{Estimate-Norm}
    (\mu_1, f_1, m)$
    \label{line:check consistency estimate ip 1}
  \State $\tilde c_2^2 \gets \textsc{Estimate-Norm}
    (\mu_2, f_2, m)$
    \label{line:check consistency estimate ip 2}
  \State $\tilde p \gets \textsc{Estimate-IP}
    (\mu_1, \mu_2, f_1, f_2, m)$
    \If{$\tilde p < \frac{C^2}{6}\epsilon^2$} \textbf{reject}
    \label{line:coarse estimate of angle}
  \EndIf
  \If{$\frac{\tilde p}{\sqrt{ \tilde c_1^2 \tilde c_2^2}} < 1-2\epsilon_1^2$}
    \textbf{reject}
  \EndIf
\EndIf
\State If $\textsc{Check-Threshold}((\mu_1)_\pi, (\mu_2)_\pi, (\tilde v_1+1)/2,
  (\tilde v_2 + 1)/2, \epsilon, \delta/6)$ accepts, \textbf{accept};
  else \textbf{reject}
\end{algorithmic}
\end{algorithm}
\begin{lemma}
\label{lemma:check consistency}
Let $\mu_1, \mu_2$ be bounded RI distributions over $\bR^n$, with shared
constant $C$, unknown to the algorithm.  Assume that for $\tau_1 =
\bE_{\mu_1}[x_1^2], \tau_2 = \bE_{\mu_2}[x_2^2]$ that $\tau_1\tau_2=1$. Let
$\epsilon, \delta \in (0,1/2)$ and let $f_1, f_2 : \bR^n \to \pmset$ be $\mu_1$-
and $\mu_2$-measurable functions.  Then \textsc{Check-Consistency} satisfies the
following:
\begin{itemize}
  \item If there exists a halfspace $h$ such that $\dist_{\mu_1}(f_1, h),
    \dist_{\mu_2}(f_2, h) < \epsilon/6$ and $h$ is aligned with $f_1, f_2$, the
    algorithm will accept with probability at least $1-\delta$;
  \item If there exist halfspaces $h_1, h_2$ with thresholds $t_1, t_2$ such
    that $\dist_{\mu_1}(f_1,h_1), \dist_{\mu_2}(f_2, h_2) < \epsilon$ and $h_1,
    h_2$ are aligned with $f_1, f_2$ with $\bE_{\mu_1}{h_1} = \bE_{\mu_1}{f_1},
    \bE_{\mu_2}{h_2} = \bE_{\mu_2}{f_2}$, then with probability at least
    $1-\delta$: if the algorithm accepts then there exists a threshold $t$ such
    that $\hdist_{\mu_1}(t_1, t), \hdist_{\mu_2}(t_2, t) < 2\epsilon$, and one
    of the following holds:
    \begin{enumerate}
      \item $\alpha(h_1, h_2) \leq \frac{\pi}{2}\epsilon$, or
      \item $\abs{\bE_{\mu_1}{f_1}} \geq 1-\frac{3}{2}\epsilon$ or
        $\abs{\bE_{\mu_2}{f_2}} \geq 1-\frac{3}{2}\epsilon$.
    \end{enumerate}
    In particular, as a consequence, there exists a halfspace $h$ such that
    $\dist_{\mu_1}(f_1, h), \dist_{\mu_2}(f_2, h) < 5\epsilon$.
  \item The algorithm requires at most $\BigO{\frac{\sqrt
    n}{\epsilon^4}\log(1/\delta)}$ random samples from each distribution.
\end{itemize}
\end{lemma}
Finally we show that we can extract a single consistent halfspace from a set of
near-halfspaces by applying \textsc{Check-Consistency} to each pair of
functions:
\begin{lemma}
\label{lemma:check consistency group}
  Let $\epsilon, \delta > 0$, let $\mu_1, \dots, \mu_k$ be a set of bounded RI
  distributions over $\bR^n$ and let $f : \bR^n \to \pmset$ be a function
  measurable by all these distributions. Suppose that for all $i \in [k]$ there
  exists a halfspace $h_i$ satisfying
    $h_i$ is aligned with $f$ on $\mu_i$,
    $\bE_{\mu_i}{h_i} = \bE_{\mu_i}{f}$, and
    $\dist_{\mu_i}(f,h_i) \leq \epsilon$.
  Assume that for all $i,j \in [k], i \neq j$,
  $\textsc{Check-Consistency}(\mu_i, \mu_j, f, f, \epsilon/5, \delta)$ has
  succeeded. Then there exists a halfspace $h$ such that for all $i \in [k],
  \dist_{\mu_i}(f,h) < 3\epsilon$.
\end{lemma}

\subsection{The \textsc{RI-Tester} Algorithm}
\label{subsection:ri tester}
The \textsc{RI-Tester} algorithm will divide the whole RI space into several
``rings'', each of which is simple. We will not take fresh samples for each run of a
subroutine inside a ring: instead we will take a large set of samples at the
start and ignore any rings which do not receive enough samples. We then simulate
each subroutine, reusing these samples each time and using the union bound to
ignore that the samples are not independent. The following lemma, whose proof is
delayed until after the presentation of the algorithm, allows us to ignore the
rings with few samples:
\begin{lemma}
\label{lemma:ignore small sections}
Let $\epsilon, \delta > 0$ and
let $\mu$ be an arbitrary probability distribution over a space $\cX$ and
suppose $\cX$ is partitioned into sets $\{ R_i \}_{i \in [k]}$ for some $k$. Let
$\{ m_i \}_{i \in [k]}$ be a set of arbitrary nonnegative numbers with maximum
value $m = \max_i m_i \geq \frac{k}{\epsilon}\ln(2k/\delta)$. Let $X$ be
a set of $M = \frac{2k}{\epsilon}m$ random samples from $\mu$. Write $A = \{ i
\in [k] : \#(X \cap R_i) < m_i \}$. Then with probability at least $1-\delta$ we
have $\sum_{i \in A} \mu(R_i) < \epsilon$.
\end{lemma}

In the application of this lemma, the number of samples $m_i$ for each ring is
determined by the number of samples required by the \textsc{Simple-Tester} and
\textsc{Check-Consistency} algorithms:
\begin{definition}
  For $\epsilon, \delta \in (0,1)$, the combined sample complexity of
  \textsc{Simple-Tester} and \textsc{Check-Consistency} with parameters
  $\epsilon,\delta$ is
  \[
  m(\epsilon,\delta) = \BigTheta{\frac{\sqrt
  n}{\epsilon^6}\log(1/\delta)}
  \]
  by Theorem \ref{thm:simple tester} and Lemma \ref{lemma:check consistency}.
  (\textsc{Check-Consistency} requires samples from two distributions, but we
  only count the number of samples required from one of them).
\end{definition}

\begin{algorithm}[h]
\label{alg:ri tester}
\caption{\textsc{RI-Tester}$(\mu, \epsilon, \delta)$}
\begin{algorithmic}[1]
\Statex \textit{Find the pivot and define the parameters:}
\State $k \gets \ceil{ \log\fracb{2n}{\sqrt{2\pi} \epsilon} }$; \qquad
       $K \gets {k \choose 2} + k + 5$
\State $\tilde t \gets \textsc{Find-Pivot}(\mu, f, \epsilon^2, \delta/K)$

\Statex \vspace{1em}
  \textit{Gather the samples and find the ``significant'' rings:}
\State Draw $X \sim \mu^{M}$ for $M = \frac{2k}{\epsilon}m(\epsilon, \delta/K)$
\State $\forall i \in [k] : (f_i, R_i, S_i, \mu_i)
  \gets \left(f, \{x : 2^{i-1} \leq (\norm x / \tilde t) < 2^i \},
    X \cap R_i, \mu \mid_{R_i} \right)$
\State Let $\mu_{k+1}, f_{k+1}$ be the projection for $f$ on $\mu
  \mid_{(T, \infty)}$, defined by Lemma \ref{lemma:projecting
  balanced functions}.
\State $A \gets \{ i \in [k+1] : \#S_i \geq m(\epsilon,\delta/K) \}$

\Statex \vspace{1em}
  \textit{Ensure the function is close to a halfspace on each ring:}
\For{$i \in A$}
  \State Simulate $\textsc{Simple-Tester}
    (\mu_i, f_i, \epsilon, \delta/K)$ on $S_i$;
    \textbf{reject} if it rejects.
    \label{line:ri simple tester}
\EndFor

\Statex \vspace{1em}
  \textit{Check that the halfspaces for each ring are consistent with each
  other:}
\For{$i,j \in A, i<j$}
  \State Simulate $\textsc{Check-Consistency}
    (\mu_i, \mu_j, f_i, f_j, \epsilon, \delta/K)$
    on $S_i$; \textbf{reject} if it rejects.
    \label{line:ri check consistency}
\EndFor

\Statex \vspace{1em}
  \textit{Ensure all rings of large radius are close to balanced:}
  \For{$i \in A : i-2 \geq \log\fracb{\sqrt n}{\sqrt{2\pi}\epsilon}$}
  \If{$\abs{\textsc{Estimate-Mean}(\mu_i,f_i,\epsilon/2,\delta/K)} > \epsilon$}
    \textbf{reject}
    \label{line:ri check not constant}
    \label{line:ri check balanced}
  \EndIf
\EndFor

\Statex \vspace{1em}
  \textit{Ensure the central section is close to the correct constant:}

\If{$\#\{x \in X : \norm x < \tilde t\} \geq m_0 \define
\BigO{\frac{1}{\epsilon^2}\log(K/\delta)}$ and $A \neq \emptyset$}
  \State $\tilde v_1 \gets
    \textsc{Estimate-Mean}(\mu_0, f, 2\epsilon/3, \delta/K)$
  \State $\tilde v_2 \gets
    \textsc{Estimate-Mean}(\mu_{[\tilde t, \infty)}, f, 2\epsilon/3, \delta/K)$
  \If{$\textsc{Check-Threshold}((\mu_0)_\pi, (\mu_{[\tilde t, \infty)})_\pi,
  (\tilde v_1+1)/2, (\tilde v_2 +1)/2, \epsilon, \delta/K)$ rejects}
    \textbf{reject}
    \label{line:ri check threshold}
  \EndIf
\EndIf
\State \textbf{accept}
\end{algorithmic}
\end{algorithm}

\textbf{Theorem \ref{thm:ri tester}.}
\textit{%
Let $\mu$ be any RI distribution over $\bR^n$, $f : \bR^n \to \pmset$ be any
measurable function, and let $\epsilon, \delta > 0$. Then
\textsc{RI-Tester}, using no knowledge of $\mu$ and at most $\BigO{\frac{\sqrt{n}
\log(n/\epsilon)}{\epsilon^7}\log\fracb{\log(n/\epsilon)}{\delta}}$ random
examples, satisfies the following:
\begin{enumerate}
\item If $f$ is a halfspace, then \textsc{RI-Tester} accepts with probability at
least $1-\delta$; and
\item If $\dist_\mu(f,h) > \epsilon$ for all halfspaces $h$ then
\textsc{RI-Tester} rejects with probability at least $1-\delta$.
\end{enumerate}}
\begin{proof}
  We remark that the \textsc{Simple-Tester} subroutine is called on
  distributions not necessarily bounded by $R=\sqrt n$; this is easily
  remedied by rescaling the samples, which we may do since the algorithm knows
  the bounds on each distribution $\mu_i$.

  We will prove the sample complexity and show that the algorithm succeeds with
  probability at least $1-\delta$.  First we note that the set $X$ of samples
  has size
  \[
  M =
  \frac{2k}{\epsilon}m(\epsilon,\delta/K) = \BigO{\frac{k \sqrt
  n}{\epsilon^7}\log(K/\delta)}
  = \BigO{\frac{\sqrt n \log(n/\epsilon)}{\epsilon^7}
  \log\fracb{\log(n/\epsilon)}{\delta}} \,.
  \]
  The subroutines which do not use the samples in $X$ are \textsc{Find-Pivot},
  \textsc{Estimate-Mean}, and \textsc{Check-Threshold}, which together require
  $\BigO{\frac{1}{\epsilon^2} \log(K/\delta)} = \BigO{\frac{1}{\epsilon^2}
  \log\fracb{\log(n/\epsilon)}{\delta}}$ samples; thus the total number of
  samples is $\BigO{\frac{\sqrt n \log(n/\epsilon)}{\epsilon^7}
  \log\fracb{\log(n/\epsilon)}{\delta}}$. We bound the probability that the
  algorithm fails by the probability that any of the subroutines fail, or if too
  few samples occur in each significant ring (i.e.~$\sum_{i \notin A} \mu(R_i)
  \geq \epsilon$). To bound the probability of this latter event, we use Lemma
  \ref{lemma:ignore small sections} on the set of $k+1$ rings $R_i$, with the
  parameter $m(\epsilon, \delta/K)$ for each ring. For this we must have
  $m(\epsilon,\delta/K) \geq \frac{k+1}{\epsilon}\ln(2kK/\delta)$: since $k =
  \log\fracb{2n}{\sqrt{2\pi}\epsilon}$ we have
  \begin{equation}
  \label{eq:ri m large enough}
    m(\epsilon, \delta/K) \geq \frac{\sqrt n}{\epsilon^6}\ln(K/\delta)
    = \omega \left(\frac{k}{\epsilon}\ln(kK/\delta)\right) \,,
  \end{equation}
  so with probability at least $1-\delta/K, \sum_{i \notin A} \mu(R_i) <
  \epsilon$. Assuming this is successful, we can bound the failure probability
  of the rest of the algorithm: There are at most ${k \choose 2} + k + 4$
  subroutines that succeed with probability at least $1-\delta/K$ so by the
  union bound the probability of success is at least $1-\delta$. For the rest of
  the proof, assume that all procedures have succeeded.

  \textbf{Completeness:} Suppose $f$ is a halfspace with threshold $t$.
  The next claim bounds the mean for all large rings, which we will use to show
  that \textsc{Simple-Tester} works on the projected function $f_{k+1}$ and that
  line \ref{line:ri simple tester} will pass.
  \begin{claim}
  \label{claim:mean of large rings}
    Set $i^* = \log\fracb{\sqrt n}{\sqrt{2\pi}\epsilon}$.
    For all $i \in [k+1], \abs{\bE_{\mu_i}{f}} \leq \epsilon/2^{i-i^*-1}$.
  \end{claim}
  \begin{proof}
    Let $r = 2^{i-1}\tilde t = 2^{i-i^*-1} 2^{i^*} \tilde t = 2^{i-i^*-1} \tilde
    t \frac{\sqrt n}{\sqrt{2\pi}} \epsilon$.  We have $\abs{\bE_{\mu_i}{f}} \leq
    \abs{\bE_{\sigma_r}{f}}$ since $f$ is a halfspace and $\mu_i$ is a convex
    combination of spheres of radius at least $\tilde t \cdot 2^{i-1}$. Now by
    Proposition \ref{prop:halfspace mean}, the bound on the density
    Proposition \ref{prop:density of sphere}, and the fact that $|t| \leq \tilde
    t$ by the guarantee of \textsc{Find-Pivot} (Lemma \ref{lemma:find pivot}),
    \[
    \abs{\bE_{\sigma_r}{f}} = \Pru{x \sim \sigma_r}{-|t| < x_1 < |t|}
    \leq 2|t| \frac{\sqrt n}{\sqrt{2\pi} r}
    \leq 2 \tilde t \frac{\sqrt n}{\sqrt{2\pi} r}
    = \epsilon/2^{i-i^*-1} \,. \qedhere
    \]
  \end{proof}

  \emph{Lines 7,8:} The examples of functions $f_1, \dotsc, f_k$ are not
  modified, so these lines pass for $i \in [k]$ by Theorem \ref{thm:simple
  tester}. It remains to show that $f_{k+1}$ passes, since these examples have
  been normalized.
  Let $c$ be the constant in the tolerance guarantee of Theorem \ref{thm:simple
  tester} (for $\eta=1$).  We have $(k+1) - i^* -1 \geq \log(2n) - \log(\sqrt n)
  = \log(2\sqrt n)$ so the mean of the function on the $\nth{(k+1)}$ ring is at
  most $\abs{\Exuc{}{f(x)}{\norm x \geq 2^k \tilde t}} \leq \epsilon/2\sqrt n <
  c\epsilon^4$ when $\epsilon \geq \frac{1}{(4c^2n)^{1/6}}$. Thus by Lemma
  \ref{lemma:projecting balanced functions} we know that $f_{k+1}$ is
  $c\epsilon^4$-close to a halfspace on the projected distribution $\mu_{k+1}$
  so line \ref{line:ri simple tester} will pass.

  \emph{Lines 9,10:} The functions $f_1, \dotsc, f_k$ pass this test by Lemma
  \ref{lemma:check consistency}.
  We have just proved that $f_{k+1}$ is $c\epsilon^4$-close to a halfspace, and
  again by Lemma \ref{lemma:projecting balanced functions} we know that
  $f_{k+1}$ is aligned with $f$, so it satisfies the requirement of
  \textsc{Check-Consistency} and will also pass these lines. 

  \emph{Lines 11,12:} Claim \ref{claim:mean of large rings} proves that
  $\abs{\bE_{\mu_i}{f}} \leq \epsilon/2$ for each $i$ on which line \ref{line:ri
  check balanced} is applied, so these lines pass.

  \emph{Line 16:}
  We have $\abs{\tilde v_1 - \bE_{\mu_0}{f}} < \epsilon, \abs{\tilde v_2 -
  \bE_{\mu_{[\tilde t, \infty)}}{f}} < \epsilon$. Let $t_1, t_2$ be thresholds
  such that the halfspaces $h_i(x) \define \sign(\inn{w,x} - t_i)$ satisfy
  $\bE_{\mu_0}{h_1} = \tilde v_1, \bE_{\mu_{[\tilde t, \infty)}}{h_2} = \tilde
  v_2$. For $i \in \{1,2\}$ and distributions $\mu_0, \mu_{[\tilde t, \infty)}$
  respectively, and assuming without loss of generality that $t \leq t_i$, we
  have
  \[
    2\epsilon/3 \geq \abs{\tilde v_i - \Ex f}
    = \abs{\Pr{z \geq t_i} - \Pr{z \geq t} - \Pr{z < t} + \Pr{z < t_i}}
    = 2\Pr{t \leq z \leq t_i}
    = 2\hdist(t, t_i) \,,
  \]
  so \textsc{Check-Threshold} accepts by Lemma \ref{lemma:check threshold}.
  
  \textbf{Soundness:} Suppose $f$ is accepted by the algorithm. We want to find
  a halfspace $h$ such that $f$ is $\epsilon$-close to $h$.

  \begin{claim}
    Suppose the algorithm succeeds and accepts $f$, and let $\tilde t$ be the
    value of \textsc{Find-Pivot}. Then there exists a halfspace $h$ such that
    \[
      \mu_\circ[\tilde t,\infty) \cdot \dist_{[\tilde t, \infty)}(f, h) < 5\epsilon
      \,.
    \]
  \end{claim}
  \begin{proof}
    Since line \ref{line:ri simple tester} passes for each $i \in A$, we know
    from Theorem \ref{thm:simple tester} that there exists a halfspace $h_i$
    aligned with $f$ on $\mu_i$ satisfying $\bE_{\mu_i}{h_i} = \bE_{\mu_i}{f}$
    and $\dist_{\mu_i}(f,h_i) \leq \epsilon$. Thus from Lemma \ref{lemma:check
    consistency group}, there exists a halfspace $h$ such that
    $\dist_{\mu_i}(f,h) < 3\epsilon$ for each $i \in A$.

    Special attention is required for $\mu_{k+1}$, which is the projection of
    the distribution over radii $[T,\infty)$ onto a sphere, as in Lemma
    \ref{lemma:projecting balanced functions}. Write $\mu_{k+1}'$ for the
    distribution $\mu$ restricted to $R_{k+1}$ (i.e.~the distribution
    $\mu_{k+1}$ before it was projected onto a sphere). If $k+1 \in A$ then we
    know the projected function $f_{k+1}$ is $3\epsilon$-close to the halfspace
    $h$ on $\mu_{k+1}$. By line \ref{line:ri check balanced}, we know that
    $\abs{\bE_{\mu_{k+1}}{f_{k+1}}} < \epsilon$, so by the lemma we know that
    $f$ is $2\epsilon$-close to the balanced halfspace $h'$ aligned with $h$; now
    we need to make sure $h'$ is close to $h$.

    Since $\mu_{k+1}'$ is the distribution over radii greater than $T$, the
    distance between $h'$ and $h$ on $\mu_{k+1}'$ is at most the distance on the
    sphere of radius $T$. We want to show, for the sphere of radius $T$,
    that the halfspace $h$ should have threshold $t$ satisfying
    \begin{equation}
    \label{eq:ri tester soundness}
      \dist_{\mu_{k+1}'}(h,h') = \Pr{0 \leq x_i < t} \leq 4\epsilon \,.
    \end{equation}
    Recall that $r = 2 \tilde t \frac{\sqrt n}{\sqrt{2\pi} \epsilon}$.

    \textbf{Case 1:} Suppose there exists $i \in A$ such that $\tilde t
    \cdot 2^{i-1} \geq r$, so \ref{line:ri check not constant} is executed for
    some such $i$.

    Then we have $\abs{\bE_{\mu_i}{f}} \leq \epsilon + \epsilon/2 \leq
    2\epsilon$ and $\dist_{\mu_i}(f,h) < 3\epsilon$ so $\abs{\bE_{\mu_i}{h}} <
    8\epsilon$. Since $\mu_{k+1}'$ is a distribution over spheres of larger
    radius than sphere $i$, we know that $\abs{\bE_{\mu_{k+1}}{h}} = 2\Pr{0 \leq
    x_1 < t} < 8\epsilon$, which proves equation \eqref{eq:ri tester soundness}.

    \textbf{Case 2:} Suppose that for all $i \leq k$ such that $\tilde t \cdot
    2^{i-1} \geq r, i \notin A$. Then lines \ref{line:ri simple tester} and
    \ref{line:ri check consistency} were not executed for any such rings $i$, so
    we may ignore them. Suppose $t > r$; then $h$ is constant on all radii at
    most $r$, and we can consider the halfspace with threshold $r$ instead of
    $t$ without changing the distance on those radii. Since $f$ is close to
    balanced on $\mu_{k+1}$, we can again consider threshold $r$. Thus we can
    assume that $t \leq r$. Since $t \leq r$ we have, for the sphere of radius
    $T$,
    \[
      \abs{\bE_{\mu_{k+1}}{h}}
      = 2\Pr{0 \leq x_1 < t}
      \leq 2t \frac{\sqrt n}{\sqrt{2\pi} T}
      \leq 2r \frac{\sqrt n}{\sqrt{2\pi} T}
      \leq 2\epsilon
    \]
    since $r = 2 \tilde t \frac{\sqrt n}{\sqrt{2\pi} \epsilon} \leq
    \epsilon T \frac{\sqrt{2\pi}}{\sqrt n}$ when $T \geq 2 \tilde t
    \frac{n}{\epsilon^2}$. Then equation \eqref{eq:ri tester soundness} holds.

    We have $m(\epsilon,\delta/K) =
    \omega\left(\frac{k}{\epsilon}\ln(kK/\delta)\right)$ (equation \eqref{eq:ri
    m large enough}), so from Lemma \ref{lemma:ignore small sections}, we know
    that $\sum_{i \notin A} \mu(R_i) < \epsilon$.  Since $\bigcup_i R_i = \{ x
    \in \bR^n : \norm{x} \geq \tilde t \}$, we conclude
    \[
    \mu_\circ[\tilde t, \infty) \cdot \dist_{[\tilde t, \infty)}(f,h)
    = \sum_{i \in A} \mu(R_i) \dist_{\mu_i}(f,h)
      + \sum_{i \notin A} \mu(R_i) \dist_{\mu_i}(f,h)
      \leq 4\epsilon + \epsilon \,. \qedhere
    \]
  \end{proof}

  Now if $\mu_\circ[0, \tilde t) < \epsilon$ we may ignore this interval. If
  $\mu_\circ[0, \tilde t) \geq \epsilon$ while $\mu_\circ[\tilde t, \infty) <
  \epsilon$ then we may ignore the latter interval and by \textsc{Find-Pivot} we
  know that for some constant $b$,
  \[
    \Pr{f(x) \neq b} \leq \Pr{f(x) \neq b \wedge \norm{x} < \tilde t} + \epsilon
    \leq \epsilon^2 + \epsilon \,.
  \]
  Thus $f$ is $2\epsilon$-close to the constant function $b$. So all that
  remains is the case where neither interval $[0, \tilde t)$ or $[\tilde t,
  \infty)$ can be safely ignored. In this case line \ref{line:ri check
  threshold} is executed. From \textsc{Find-Pivot} we have
  \[
    \epsilon \cdot \Pruc{f(x) \neq b}{\norm x < \tilde t}
    \leq \Pr{f(x) \neq b \wedge \norm x < \tilde t}
    < \epsilon^2
  \]
  so $f$ is $\epsilon$-close to the constant $b$ on $\mu_0$. Let $h_1$ be the
  halfspace aligned with $h$ with $\bE_{\mu_0}{h_1} = \tilde v_1$ and let
  $h_2$ be the halfspace aligned with $h$ with $\bE_{\mu_{[\tilde t,
  \infty)}}{h_2} = \tilde v_2$. Since $f$ is $\epsilon$-close to the constant
  $b$ on $\mu_0$ we have
  \[
    \dist_{\mu_0}(f, h_1) \leq \dist_{\mu_0}(f, b) + \dist_{\mu_0}(b, h_1)
    \leq \epsilon + \frac{1}{2}\abs{\Exuc{}{b - h_1(x)}{\norm x < \tilde t}}
    \leq 2\epsilon \,,
  \]
  since $\abs{\Exuc{}{b-h_1(x)}{\norm x < \tilde t}} \leq \abs{1 - \tilde v_1}
  \leq 2\epsilon$. We can show a similar bound for $h_2$:
  \begin{align*}
    \dist_{[\tilde t, \infty)}(f, h_2)
    &\leq \dist_{[\tilde t, \infty)}(f, h) + \dist_{[\tilde t, \infty)}(h, h_2) \\
    &\leq \epsilon + \frac{1}{2}\abs{\Exu{[\tilde t, \infty)}{h - h_2}} \\
    &\leq 2\epsilon + \frac{1}{2}\abs{\Exu{[\tilde t, \infty)}{f - h_2}} \\
    &\leq 3\epsilon \,.
  \end{align*}
  Now since $h_1, h_2$ are both aligned with $h$, from \textsc{Check-Threshold}
  (Lemma \ref{lemma:check threshold}) we are guaranteed that there exists a
  halfspace $h'$ such that $\dist_{[0, \tilde t)}(h_1, h'), \dist_{[\tilde t,
  \infty)}(h_2, h') \leq \epsilon$. Therefore we have
  \begin{align*}
  \dist(f, h')
  &\leq \mu_\circ[0, \tilde t) \cdot \dist_{[0, \tilde t)}(f, h')
    + \mu_\circ[\tilde t, \infty) \cdot \dist_{[\tilde t, \infty)}(f, h') \\
  &\leq \mu_\circ[0, \tilde t)(2\epsilon + \epsilon)
    + \mu_\circ[\tilde t, \infty)(3\epsilon + \epsilon)
  \leq 4\epsilon \,. \qedhere
  \end{align*}
\end{proof}

\ignore{
\note{Read \cite{NS60, Kla94} to see if there is any improvement on the
following lemma. This is a variant ``Double Dixie Cup Problem''. (These papers
seem to corroborate the below result at first glance, i.e. that for the uniform
distribution the expected number of samples required is
$ k\ln k + (m-1)\ln \ln k + C_m + o(1)$ for some constant $C_m$.)}
}

\textbf{Lemma \ref{lemma:ignore small sections} (restated).}
Let $\epsilon, \delta > 0$ and
let $\mu$ be an arbitrary probability distribution over a space $\cX$ and
suppose $\cX$ is partitioned into sets $\{ R_i \}_{i \in [k]}$ for some $k$. Let
$\{ m_i \}_{i \in [k]}$ be a set of arbitrary nonnegative numbers with maximum
value $m = \max_i m_i \geq \frac{k}{\epsilon}\ln(2k/\delta)$. Let $X$ be
a set of $M = \frac{2k}{\epsilon}m$ random samples from $\mu$. Write $A = \{ i
\in [k] : \#(X \cap R_i) < m_i \}$. Then with probability at least $1-\delta$ we
have $\sum_{i \in A} \mu(R_i) < \epsilon$.

\begin{proof}
  Write $m = \max_i m_i$.
  Define a parameter $L$, to be fixed later, and write $M =
  (L+1)\frac{k}{\epsilon}m$. Write $S_i = \#( X \cap R_i )$. Let $i \in [k]$ be
  any index such that $\mu(R_i) \geq \epsilon/k$. Then we want to show that
  \[
    \Pr{ S_i < m_i } \leq \delta/k \,.
  \]
  We have $S_i = \sum_{x \in X} \ind{x \in R_i}$ so $S_i$ is a sum of
  independent Bernoulli random variables with expectation $\Ex{S_i} = M
  \mu(R_i)$. Then
  \[
  \Pr{ S_i < m_i } \leq \Pr{\abs{M\mu(R_i) - S_i} > M\mu(R_i) - m_i}
  \]
  and
  \[
  M\mu(R_i) = \mu(R_i) (L+1)\frac{km}{\epsilon}
  \geq (L+1) \frac{\epsilon}{k} \frac{km}{\epsilon} = (L+1)m \,.
  \]
  Then $M\mu(R_i) - m_i \geq (L+1)m - m_i \geq Lm$ so
  \[
  \Pr{S_i < m_i}
  \leq \Pr{\abs{M\mu(R_i) - S_i} > M\mu(R_i) - m_i}
  \leq \Pr{\abs{M\mu(R_i) - S_i} > Lm}
  \leq 2\exp{-\frac{2 (Lm)^2}{M}}
  \]
  by Hoeffding's inequality. We want this probability to be at most $\delta/k$
  so we want
  \[
    \frac{2(Lm)^2}{M} \geq \ln(2k/\delta) \,.
  \]
  Substituting
  $Lm = \left(\frac{\epsilon}{km}M - 1 \right)m = M\frac{\epsilon}{k} - m$, we
  have
  \[
  \frac{2(Lm)^2}{M} = \frac{2(M\epsilon/k - m)^2}{M}
  = 2M\frac{\epsilon^2}{k^2} - 4\frac{\epsilon}{k} m + 2\frac{m^2}{M}
  \]
  Suppose $m \geq \frac{k}{\epsilon}\ln(2k/\delta)$. Then for $M =
  \frac{2m}{\epsilon}k$ the first two terms cancel out and the expression is
  \[
  \frac{2(Lm)^2}{M} = \frac{m\epsilon}{k} \geq \ln(2k/\delta) \,,
  \]
  which is what we want. Using the union bound, we get that the probability that
  any set $R_i$ with $\mu(R_i) \geq \epsilon/k$ has $S_i < m_i$ is at most
  $\delta$:
  \[
  \Pr{\exists i : \mu(R_i) \geq \epsilon/k, S_i < m_i}
  \leq k \Pr{S_i < m_i} \leq \delta \,.
  \]
  Assuming this event does not occur, the only indices $i$ satisfying $S_i <
  m_i$ are those with $\mu(R_i) < \epsilon/k$, so
  \[
    \sum_{i \in A} \mu(R_i) < \sum_{i \in A} \epsilon/k \leq \epsilon \,.
    \qedhere
  \]
\end{proof}

\subsection{Proofs for Checking Consistency}
\label{subsection:proofs of consistency checks}
We now present the proofs of Lemmas \ref{lemma:check threshold} and
\ref{lemma:check consistency}. For convenience we will restate these lemmas:

\textbf{Lemma \ref{lemma:check threshold}}
Let $\epsilon,\delta \in (0,1/2)$, $v_1, v_2 \in [0,1]$ and let $\mu_1, \mu_2$
be any distributions over $\bR$, unknown to the algorithm. Let $t_1, t_2$
satisfy $\Pru{x \sim \mu_1}{x_1 \geq t_1} = v_1, \Pru{y \sim \mu_2}{y_1 \geq
t_2} = v_2$. Then, using $\BigO{\frac{1}{\epsilon^2}\log(1/\delta)}$ samples,
with probability at least $1-\delta$ \textsc{Check-Threshold} satisfies the
following:
\begin{enumerate}
\item If there exists a threshold $t$ such that $\hdist_{\mu_1}(t_1, t),
\hdist_{\mu_2}(t_2, t) < \epsilon/3$ then \textsc{Check-Threshold} accepts; and
\item If \textsc{Check-Threshold} accepts then there exists $t$ such that
$\hdist_{\mu_1}(t_1, t), \hdist_{\mu_2}(t_2, t) \leq \epsilon$.
\end{enumerate}
\begin{proof}[Proof of Lemma \ref{lemma:check threshold}]
  First we show that $\epsilon/3 \leq \mu_1(\tilde a_1,
  t_1],\mu_1(t_2, \tilde b_1], \mu_2(\tilde a_2, t_2], \mu_2(t_2,\tilde b_2]
  \leq \epsilon$ with high probability. This is accomplished by using
  Proposition \ref{prop:estimate threshold} 4 times with parameters $\epsilon/3,
  \delta/4$ and each combination
  of $v = v_i \pm 2\epsilon/3$ for $i$ and $\pm$.

  Suppose without loss of generality that $t_1 \leq t_2$, and assume the above
  estimations succeed. Note that if $\tilde a_i = -\infty$ then $v_i +
  2\epsilon/3 > 1$ so $v_i > 1-2\epsilon/3 \geq 2/3$. Therefore $\tilde b_i$ is
  finite, and the opposite direction holds as well, so only one of $\tilde a_i,
  \tilde b_i$ is infinite. If $\tilde a_1$ is finite then $\mu_1[\tilde a_1,
  t_1) \geq \epsilon/3$ so $\tilde a_1 < t_1$; if $\tilde a_1=-\infty$ this
  clearly holds as well. Similarly $t_2 < \tilde b_2$, so $\tilde a_1 \leq t_1
  \leq t_2 \leq \tilde b_2$.
  
  For the first guarantee, we have $\mu_1[t_1,t), \mu_2[t,t_2) <
  \epsilon/3$, and we want to show that the intervals $[\tilde a_1, \tilde b_1],
  [\tilde a_2,\tilde b_2]$ overlap. If $t > \tilde b_1$ then
  \[
  \mu_1[t_1,t)
  \geq \mu_1[t_1, \tilde b_1)
  \geq \epsilon/3
  \]
  which is a contradiction. Likewise if $t < \tilde a_2$ then $\mu_2[t,t_1) \geq
  \epsilon/3$; thus $\tilde a_2 \leq t \leq \tilde b_1$. Therefore the intervals
  overlap and the test passes.

  Now suppose that the test passes, so the intervals $[\tilde a_1, \tilde b_1],
  [\tilde a_2, \tilde b_2]$ overlap.  Let $t$ be some point inside the
  intersection. We need to show that $\mu_1[t_1,t],\mu_2[t,t_2] \leq \epsilon$,
  which is easily seen since $\mu_1[t_1,t] \leq \mu_1[\tilde a_1, \tilde b_1]
  \leq \epsilon$, and the analogous argument holds for $\mu_2$.
\end{proof}

\textbf{Lemma \ref{lemma:check consistency}.}
Let $\mu_1, \mu_2$ be bounded RI distributions over $\bR^n$, with shared
constant $C$, unknown to the algorithm. Assume that for  $\tau_1 =
\bE_{x \sim \mu_1}[x_1^2], \tau_2 = \bE_{x \sim \mu_2}[x_1^2]$ that $\tau_1 \tau_2 = 1$. Let
$\epsilon, \delta \in (0,1/2)$ and let $f_1, f_2 : \bR^n \to \pmset$ be $\mu_1$-
and $\mu_2$-measurable functions.  Then \textsc{Check-Consistency} satisfies the
following:
\begin{itemize}
  \item If there exists a halfspace $h$ such that $\dist_{\mu_1}(f_1, h),
    \dist_{\mu_2}(f_2, h) < \epsilon/6$ and $h$ is aligned with $f_1, f_2$, the
    algorithm will accept with probability at least $1-\delta$;
  \item If there exist halfspaces $h_1, h_2$ with thresholds $t_1, t_2$ such
    that $\dist_{\mu_1}(f_1,h_1), \dist_{\mu_2}(f_2, h_2) < \epsilon$ and $h_1, h_2$ are
    aligned with $f_1, f_2$ with $\bE_{\mu_1}{h_1} = \bE_{\mu_1}{f_1},
    \bE_{\mu_2}{h_2} = \bE_{\mu_2}{f_2}$,
      then with probability at least $1-\delta$: if the algorithm accepts then
      there exists a threshold $t$ such that $\hdist_{\mu_1}(t_1, t),
      \hdist_{\mu_2}(t_2, t) < 2\epsilon$, and one of the following holds:
    \begin{enumerate}
      \item $\alpha(h_1, h_2) \leq \frac{\pi}{2}\epsilon$, or
      \item $\abs{\bE_{\mu_1}{f_1}} \geq 1-\frac{3}{2}\epsilon$ or
        $\abs{\bE_{\mu_2}{f_2}} \geq 1-\frac{3}{2}\epsilon$.
    \end{enumerate}
    In particular, as a consequence, there exists a halfspace $h$ such that
    $\dist_{\mu_1}(f_1, h), \dist_{\mu_2}(f_2, h) < 5\epsilon$.
  \item The algorithm requires at most $\BigO{\frac{\sqrt
    n}{\epsilon^4}\log(1/\delta)}$ random samples from each distribution.
\end{itemize}
\begin{proof}
  The calls to \textsc{Estimate-Mean} require at most
  $\BigO{\frac{1}{\epsilon^2}\log(1/\delta)}$ samples (Lemma \ref{lemma:estimate
  mean}), the call to \textsc{Check-Threshold} requires at most
  $\BigO{\frac{1}{\epsilon^2}\log(1/\delta)}$ samples (Lemma \ref{lemma:check
  threshold}), and the three calls to \textsc{Estimate-IP} and
  \textsc{Estimate-Norm} each require $m = \BigTheta{\frac{\sqrt
  n}{\epsilon^4}\log(1/\delta)}$ samples. So the total sample complexity is
  $\BigTheta{\frac{\sqrt n}{\epsilon^4} \log(1/\delta)}$. By the union bound,
  the probability of failure is at most $\delta$. In the following, assume all
  estimations succeed, and let $c_1 = \norm{\bE_{\mu_1}[xf_1(x)]}, c_2 =
  \norm{\bE_{\mu_2}[xf_2(x)]}, p = \inn{\bE_{\mu_1}[xf_1(x)],
  \bE_{\mu_2}[xf_2(x)]}$.
  
  \textbf{Completeness:} Suppose there exists a halfspace $h$ of distance at
  most $\epsilon/2$ to both $f_1$ and $f_2$, and let $t$ be the threshold of $h$.
  Let $t_1, t_2$ be the thresholds of the halfspaces $h_1, h_2$ aligned with
  $f_1, f_2, h$ such that $\bE_{\mu_1}{f_1} = \bE_{\mu_1}{h_1}$ and
  $\bE_{\mu_2}{f_2} = \bE_{\mu_2}{h_2}$.
  Then
  \[
  \hdist_{\mu_1}(t_1, t)
  = \dist_{\mu_1}(h_1,h)
  = \frac{1}{2}\abs{\Exu{x \sim \mu_1}{h_1(x) - h(x)}}
  = \frac{1}{2}\abs{\Exu{x \sim \mu_1}{f_1(x) - h(x)}}
  \leq \dist_{\mu_1}(f_1, h) < \epsilon/6
  \]
  and same for $\hdist_{\mu_2}(t_2, t)$. Let $\tilde t_1, \tilde t_2$ be the
  thresholds for halfspaces with volumes $\tilde v_1, \tilde v_2$, respectively,
  on $\mu_1$ and $\mu_2$. Then $\hdist_{\mu_1}(\tilde t_1, t_1) = |\tilde v_1 - v_1|/2
  \leq \epsilon/6$ and the same for $\tilde t_2$. We have for each $i$ that
  $\Pru{z \sim (\mu_i)_\pi}{z \geq \tilde t_i} = (\tilde v_i + 1)/2$ which are
  the parameters used for \textsc{Check-Threshold}. Since
  \[
    \hdist_{\mu_1}(\tilde t_1, t) \leq \hdist_{\mu_1}(\tilde t_1, t_1) + \hdist_{\mu_1}(t_1, t)
    \leq \epsilon/6 + \epsilon/6 \leq \epsilon/3 \,,
  \]
  the condition for the first guarantee of \textsc{Check-Threshold} (Lemma
  \ref{lemma:check threshold}) is satisfied, so \textsc{Check-Threshold} passes.

  Suppose $|\tilde v_1|, |\tilde v_2| < 1-\epsilon$, so that
  $\abs{\Ex{h_1}}, \abs{\Ex{h_2}} < 1-\epsilon/2$. Keep in mind that $\mu_1,
  \mu_2$ may not be isotropic, so when we compare the estimates $\tilde c_1,
  \tilde c_2$, and $\tilde p$ we will use the multiplicative guarantees of
  \textsc{Estimate-IP}. However, to prove that the multiplicative guarantees on
  each estimate hold, we scale the distributions appropriately to match the
  hypothesis of Lemma \ref{lemma:estimate ip}.  The multiplicative
  guarantee holds for $\tilde c_1$ iff it holds when $\mu_1$ is scaled to be
  isotropic, so assume this is the case. By Proposition \ref{prop:simple ri
  halfspaces have large centers} we know that
  \begin{equation}
  \label{eq:lower bound on center-norms}
  \norm{\Exu{x \sim \mu_1}{x h_1(x)}} \geq \frac{C}{2} \epsilon \,.
  \end{equation}
  Then from Lemma \ref{lemma:upper bound on norm distance}, we have for some
  constant $K > 0$,
  \begin{equation}
  \label{eq:check consistency center norms are large}
  \begin{aligned}
    c_1^2 &\geq \left(\norm{\Exu{\mu_1}{xh_1(x)}} - K\epsilon
    \sqrt{\ln(1/\epsilon)} \right)^2
    = \frac{C^2}{4}\epsilon^2 - \frac{CK}{2}\epsilon^2 \sqrt{\ln(1/\epsilon)} +
     K^2\epsilon^2 \ln(1/\epsilon) \,.
  \end{aligned}
  \end{equation}
  We may assume $K \geq 1$ since $K$ is an upper bound in Lemma \ref{lemma:upper
  bound on norm distance}, and we also know $C \leq 1$. Thus the following
  inequalities hold when $\epsilon \leq 1/2 \leq e^{-1/4}$:
  \[
    K \sqrt{\ln(1/\epsilon)} \geq C/2 \iff
    \sqrt{\ln(1/\epsilon)} \geq C/(2K) \iff
    \epsilon \leq \exp{-\frac{C^2}{4K^2}} \,,
  \]
  and this implies that $c_1^2 \geq \frac{C^2}{4}\epsilon^2$; this holds also
  for $c_2^2$ by the same argument.

  Now we must show that $\tilde p$ satisfies the multiplicative guarantee. We
  may no longer assume that both $\mu_1, \mu_2$ are isotropic since we are
  comparing the two, but we know that for $\tau_1 = \bE_{\mu_1}[x_1^2], \tau_2 =
  \bE_{\mu_2}[x_1^2]$, $\tau_1 \tau_2 = 1$. Rescaling inequality \eqref{eq:lower
  bound on center-norms} gives us $c_1 \geq \sqrt{\tau_1} \cdot C\epsilon/4$ and
  similar for $c_2$.  Since $f_1, f_2$ are aligned, $p = c_1c_2$, so we also
  have $p \geq \sqrt{\tau_1\tau_2} \cdot C^2\epsilon^2/4 = C^2\epsilon^2/4$.
  Thus
  \[
  \tilde p
  \geq \frac{C^2}{4}\epsilon^2 - \epsilon_1^2
  \geq C^2(3/12 - 1/12)\epsilon^2
  \geq \frac{C^2}{6} \epsilon^2
  \]
  since $\epsilon_1^2 = \frac{C^2}{12}\epsilon^2$. Then line \ref{line:coarse
  estimate of angle} passes and $m = \BigTheta{\frac{\sqrt
  n}{\epsilon_1^4}\log(1/\delta)} \geq \BigOmega{\frac{\sqrt n}{\epsilon_1^2
  p^2} \log(1/\delta)}$ so (for large enough constant in the definition of $m$)
  the multiplicative accuracy guarantee for \textsc{Estimate-IP} holds and we
  have $\tilde c_1^2 = (1\pm\epsilon_1^2)c_1^2, \tilde c_2^2 =
  (1\pm\epsilon_1^2)c_2^2, \tilde p = (1 \pm \epsilon_1^2)p$. Then the test
  passes since, by Claim \ref{claim:fraction lower bound} below, we have
  \[
    \frac{\tilde p}{\sqrt{\tilde c_1^2 \tilde c_2^2}}
    \geq \frac{1-\epsilon_1^2}{1+\epsilon_1^2} \frac{p}{c_1c_2}
    \geq 1-2\epsilon_1^2 \geq 1-2\epsilon^2 \,.
  \]
  \begin{claim}
  \label{claim:fraction lower bound}
    For any $x \in [0,1]$, $\frac{1-x}{1+x} \geq 1-2x$.
  \end{claim}
  \begin{proof}
  $
    \frac{d}{dx} \frac{1-x}{1+x} = - \frac{1}{1+x} + (1-x) \frac{d}{dx} (1+x)^{-1}
    = - \frac{1}{1+x} - (1-x)(1+x)^{-2}
    = - \frac{2}{(1+x)^2} \geq -2 \,.
  $
  \end{proof}

  \textbf{Soundness:} Suppose that the algorithm accepts. By the guarantee on
  \textsc{Check-Threshold} (Lemma \ref{lemma:check threshold}) there exists a
  threshold $t$ such that $\hdist_{\mu_1}(\tilde t_1, t), \hdist_{\mu_2}(\tilde t_2, t) \leq
  \epsilon$.
  \[
  \hdist_{\mu_1}(t_1, t) \leq \hdist_{\mu_1}(t_1, \tilde t_1)
    + \hdist_{\mu_1}(\tilde t_1, t) \leq \epsilon/3 + \epsilon \,,
  \]
  and the same for $t_2$, so the first conclusion holds. Now there are two cases
  for the remaining tests:
  
  \textbf{Case 1:}
  Consider the case where either $|\tilde v_1| \geq 1-\epsilon$ or $|\tilde v_2|
  \geq 1-\epsilon$. Without loss of generality assume $|\tilde v_1| \geq
  1-\epsilon$, so that $\abs{\Exu{\mu_1}{f_1}} \geq 1-\frac{3}{2}\epsilon$. We must show
  that there exists a halfspace $h$ such that $\dist_{\mu_1}(f_1, h), \dist_{\mu_2}(f_2, h)
  \leq 4\epsilon$. Let $h$ be the halfspace with threshold $t$ that is aligned
  with $h_2$. Then
  \[
  \dist_{\mu_2}(f_2, h) \leq \dist_{\mu_2}(f_2, h_2) + \dist_{\mu_2}(h_2, h)
  \leq \epsilon + \hdist_{\mu_2}(t_2, t)
  \leq 3\epsilon \,.
  \]
  Next, let $h_1'$ be the halfspace $h_1$ rotated to be aligned with $h$. Then
  \[
    \dist_{\mu_1}(f_1, h) \leq \dist_{\mu_1}(f_1, h_1) + \dist_{\mu_1}(h_1, h)
    \leq \epsilon + \dist_{\mu_1}(h_1, h_1') + \dist_{\mu_1}(h_1', h)
    \leq \epsilon + 2\epsilon + \hdist_{\mu_1}(t_1, t)
    \leq 5\epsilon
  \]
  where the second-last inequality holds because $\abs{\bE_{\mu_1}{h_1}} =
  \abs{\bE_{\mu_1}{f_1}} \geq
  1-2\epsilon$ (Proposition \ref{prop:distance from mean}).

  \textbf{Case 2:}
  Now consider the case where $|\tilde v_1|, |\tilde v_2| \leq 1-\epsilon$, so
  $\abs{\bE_{\mu_1}{f_2}}, \abs{\bE_{\mu_2}{f_2}} \leq 1-\epsilon/2$.

  By line \ref{line:coarse estimate of angle} we know $\tilde p \geq
  \frac{C^2}{6} \epsilon^2$ and using the additive error guarantee of
  \textsc{Estimate-IP} we know $p \geq \frac{C^2}{12}\epsilon^2$. Therefore $p$
  satisfies the guarantee for the multiplicative error as well, so we know
  $\tilde p = (1 \pm \epsilon_1^2)p$. Now using the same argument from the
  completeness proof (equation \eqref{eq:check consistency center norms are
  large}), we know that $c_1^2, c_2^2 \geq \Omega(\epsilon^2)$ so the conditions
  for multiplicative error for $\tilde c_1, \tilde c_2$, from Lemma
  \ref{lemma:estimate ip}, are satisfied, and $\tilde c_i = (1 \pm \epsilon_1^2)
  c_i$ for each $i$.

  Write $\beta = p / \sqrt{c_1 c_2}$ and $\tilde \beta = \tilde p / \sqrt{\tilde
  c_1^2 \tilde c_2^2}$, which by the multiplicative guarantees satisfies $\tilde
  \beta = \frac{1 \pm \epsilon_1^2}{\sqrt{(1 \pm \epsilon_1^2)(1 \pm
  \epsilon_1^2)}} \beta$.  Since the algorithm accepts, we must have $\tilde
  \beta \geq 1-2\epsilon_1^2$ so
  \begin{align*}
  \beta
  = \frac{\sqrt{(1 \pm \epsilon_1^2)(1 \pm \epsilon_1^2)}}{1 \pm \epsilon_1^2} \tilde \beta
  \geq \frac{1-\epsilon_1^2}{1+\epsilon_1^2}\tilde \beta
  &\geq (1-2\epsilon_1^2) \tilde \beta
    \xh{(Proposition \ref{claim:fraction lower bound})} \\
  &\geq (1-2\epsilon_1^2)^2
  = 1 - 4\epsilon_1^2 + 4\epsilon_1^4
  \geq 1 - 4\epsilon_1^2 \,.
  \end{align*}
  Then since $\epsilon_1 \leq \epsilon/\sqrt{12}$, the angle is at most
  \[
  \alpha(h_1, h_2) = \asin(\sqrt{1-\beta^2})
  \leq \frac{\pi}{2} \sqrt{1-\beta^2}
  \leq \frac{\pi}{2} \sqrt{1-(1-4\epsilon_1^2)^2}
  = \frac{\pi}{2}\sqrt{8\epsilon^2(1-2\epsilon_1^2)}
  \leq \frac{\pi}{2} \sqrt{8} \epsilon_1 < \frac{\pi}{2} \epsilon \,.
  \]
  Finally, we will prove that there exists a halfspace $h$ such that
  $\dist_{\mu_1}(f_1, h), \dist_{\mu_2}(f_2, h) < 5\epsilon$. Let $h$ be the
  halfspace aligned with $h_1$ with threshold $t$. Then
  \[
  \dist_{\mu_1}(f_1, h) \leq \dist_{\mu_1}(f_1, h_1) + \dist_{\mu_1}(h_1, h)
  = \dist_{\mu_1}(f_1, h_1) + \hdist_{\mu_1}(t_1, t) \leq 3\epsilon
  \]
  and
  \[
  \dist_{\mu_2}(f_2, h) \leq \dist_{\mu_2}(f_2, h_2) + \dist_{\mu_2}(h_2, h) \leq
    \epsilon + \hdist_{\mu_2}(t_2, t) + \frac{\alpha(h_2,h)}{\pi}
    \leq (4 + 1/2)\epsilon \,. \qedhere
  \]
\end{proof}

\textbf{Lemma \ref{lemma:check consistency group}.}
  Let $\epsilon, \delta > 0$, let $\mu_1, \dots, \mu_k$ be a set of bounded RI
  distributions over $\bR^n$ and let $f : \bR^n \to \pmset$ be a function
  measurable by all these distributions. Suppose that for all $i \in [k]$ there
  exists a halfspace $h_i$ such that
    $h_i$ is aligned with $f$ on $\mu_i$,
    $\bE_{\mu_i}{h_i} = \bE_{\mu_i}{f}$, and
    $\dist_i(f,h_i) \leq \epsilon$.
  Assume that for all $i,j \in [k], i \neq j$,
  $\textsc{Check-Consistency}(\mu_i, \mu_j, f, f, \epsilon/5, \delta)$ has
  succeeded. Then there exists a halfspace $h$ such that for all $i \in [k],
  \dist_i(f,h) < 3\epsilon$.

\begin{proof}
  By the guarantee on \textsc{Check-Consistency} (Lemma \ref{lemma:check
  consistency}), we have for all $i,j$ there exists a threshold $t_{i,j}$ such
  that $\hdist_i(t_i, t_{i,j}),\hdist_j(t_j, t_{i,j}) < \epsilon$ and one of the
  following holds:
  \begin{enumerate}
    \item $\alpha(h_i,h_j) \leq \frac{\pi}{2} \epsilon$, or
    \item $\abs{\bE_{\mu_i}{f}} \geq 1-\epsilon$ or
      $\abs{\bE_{\mu_j}{f}} \geq 1-\epsilon$.
  \end{enumerate}
  We will first show that there exists some threshold $t$ such that
  $\hdist_i(t_i, t) \leq \epsilon$ for all $i \in [k]$. For each $i$ define the
  interval $I(i) \define [a_i, b_i]$ where $a_i, b_i$ are
  \[
    a_i \define \min_j t_{i,j} \qquad, \qquad b_i \define \max_j t_{i,j} \,.
  \]
  $\hdist_i(a,b)$ is clearly a non-decreasing function with $|a-b|$, so for any
  $t' \in I(i)$ we have $\hdist_i(t_i, t') \leq \max\{\hdist_i(t_i, a_i),
  \hdist_i(t_i, b_i)\} \leq \epsilon$. For any $i,j, t_{i,j} \in I(i) \cap I(j)$
  so $I(i) \cap I(j) \neq \emptyset$; this implies that the set of intervals $I(i)$
  are pairwise nonempty, so $ \bigcap_i I(i) \neq \emptyset$.
  Then there is some $t$ in this intersection, and $t$ therefore satisfies
  $\hdist_i(t_i,t) \leq \epsilon$ for all $i$.

  If there exists some $z$ such that $\abs{\bE_{\mu_z}{f}} < 1-\epsilon$, let
  $h$ be the halfspace aligned with $h_z$ with threshold $t$. Otherwise let $h$
  be an arbitrary halfspace with threshold $t$.  Now let $i \in [k]$. There are
  two cases.

  \textbf{Case 1:} $\abs{\bE_{\mu_i}{f}} \geq 1-\epsilon$.
  In this case, let $h_i'$ be the halfspace $h_i$ rotated to be aligned with $h$.
  Then by Proposition \ref{prop:distance from mean}, $\dist_i(h_i,h_i') \leq
  \epsilon$, so
  \[
  \dist_i(f,h) \leq \dist_i(f,h_i) + \dist_i(h_i, h_i') + \dist_i(h_i',h)
  \leq \epsilon + \epsilon + \hdist_i(t_i,t) \leq 3\epsilon \,.
  \]

  \textbf{Case 2:} $\abs{\bE_{\mu_i}{f}} < 1-\epsilon$. In this case, there must
  be some $z$ as specified above, and $h = h_z$. Let $h_i'$ be the halfspace
  aligned with $h_z$ with threshold $t_i$. Then $\alpha(h_i,h_z) \leq
  \pi\epsilon/2$ so by Proposition \ref{prop:distance decomposition} we have $\dist_i(h_i,h_i') =
  \alpha(h_i,h_z)/\pi \leq \epsilon/2$. Thus
  \[
  \dist_i(f,h) \leq \dist_i(f,h_i) + \dist_i(h_i, h_i') + \dist_i(h_i', h)
  \leq \epsilon + \alpha(h_i,h_i')/\pi + \hdist_i(t_i, t)
  \leq 3\epsilon \,. \qedhere
  \]
\end{proof}

\section*{Acknowledgments}
Many thanks to Eric Blais for suggesting this problem and for many helpful
discussions. Thanks also to Eric Blais and Amit Levi for comments on the draft
of this paper, and to the anonymous reviewers for their careful reading and
comments.

\bibliographystyle{alpha}
\bibliography{references.bib}

\begin{thebibliography}{KKMS08}

\bibitem[BBBY12]{BBBY12}
Maria-Florina Balcan, Eric Blais, Avrim Blum, and Liu Yang.
\newblock Active property testing.
\newblock In {\em Foundations of Computer Science (FOCS), 2012 IEEE 53rd Annual
  Symposium on}, pages 21--30. IEEE, 2012.

\bibitem[BDEL03]{BEL03}
Shai Ben-David, Nadav Eiron, and Philip~M Long.
\newblock On the difficulty of approximately maximizing agreements.
\newblock {\em Journal of Computer and System Sciences}, 66(3):496--514, 2003.

\bibitem[BLM13]{BLM13}
St{\'e}phane Boucheron, G{\'a}bor Lugosi, and Pascal Massart.
\newblock {\em {Concentration Inequalities: A Nonasymptotic Theory of
  Independence}}.
\newblock Oxford university press, 2013.

\bibitem[Cho61]{Chow61}
Chao-Kong Chow.
\newblock On the characterization of threshold functions.
\newblock In {\em {Switching Circuit Theory and Logical Design, 1961. SWCT
  1961. Proceedings of the Second Annual Symposium on}}, pages 34--38. IEEE,
  1961.

\bibitem[DDFS14]{DDFS14}
Anindya De, Ilias Diakonikolas, Vitaly Feldman, and Rocco~A Servedio.
\newblock {Nearly Optimal Solutions for the Chow Parameters Problem and
  Low-weight Approximation of Halfspaces}.
\newblock {\em Journal of the ACM (JACM)}, 61(2):11, 2014.

\bibitem[DS13]{DS13}
Ilias Diakonikolas and Rocco~A Servedio.
\newblock Improved approximation of linear threshold functions.
\newblock {\em Computational Complexity}, 22(3):623--677, 2013.

\bibitem[Eld15]{Eld15}
Ronen Eldan.
\newblock A two-sided estimate for the gaussian noise stability deficit.
\newblock {\em Inventiones mathematicae}, 201(2):561--624, 2015.

\bibitem[GGR98]{GGR98}
Oded Goldreich, Shafi Goldwasser, and Dana Ron.
\newblock Property testing and its connection to learning and approximation.
\newblock {\em Journal of the ACM (JACM)}, 45(4):653--750, 1998.

\bibitem[Gol06]{Gol06}
Paul~M Goldberg.
\newblock {A Bound on the Precision Required to Estimate a Boolean Perceptron
  from its Average Satisfying Assignment}.
\newblock {\em SIAM Journal on Discrete Mathematics}, 20(2):328--343, 2006.

\bibitem[GS07]{GS07}
Dana Glasner and Rocco~A Servedio.
\newblock Distribution-free testing lower bounds for basic boolean functions.
\newblock In {\em {Approximation, Randomization, and Combinatorial
  Optimization. Algorithms and Techniques}}, pages 494--508. Springer, 2007.

\bibitem[Kea98]{Kearns98}
Michael Kearns.
\newblock Efficient noise-tolerant learning from statistical queries.
\newblock {\em Journal of the ACM (JACM)}, 45(6):983--1006, 1998.

\bibitem[KKMS08]{KKMS08}
Adam~Tauman Kalai, Adam~R Klivans, Yishay Mansour, and Rocco~A Servedio.
\newblock Agnostically learning halfspaces.
\newblock {\em SIAM Journal on Computing}, 37(6):1777--1805, 2008.

\bibitem[Lon94]{Long94}
Philip~M Long.
\newblock On the sample complexity of {PAC} learning half-spaces against the
  uniform distribution.
\newblock {\em IEEE transactions on neural networks/a publication of the IEEE
  Neural Networks Council}, 6(6):1556--1559, 1994.

\bibitem[Lon03]{Long03}
Philip~M Long.
\newblock An upper bound on the sample complexity of {PAC}-learning halfspaces
  with respect to the uniform distribution.
\newblock {\em Information Processing Letters}, 87(5):229--234, 2003.

\bibitem[LV03]{LV03}
L{\'a}szl{\'o} Lov{\'a}sz and Santosh Vempala.
\newblock Logconcave functions: Geometry and efficient sampling algorithms.
\newblock In {\em Proceedings of the 44th Annual IEEE Symposium on the
  Foundations of Computer Science}, pages 640--649. IEEE, 2003.

\bibitem[MN15]{MN12}
Elchanan Mossel and Joe Neeman.
\newblock Robust optimality of gaussian noise stability.
\newblock {\em Journal of the European Mathematical Society}, 17(2):433--482,
  2015.

\bibitem[MORS09]{MORS09}
Kevin Matulef, Ryan O'Donnell, Ronitt Rubinfeld, and Rocco~A Servedio.
\newblock Testing $\pm$1-{W}eight {H}alfspaces.
\newblock In {\em Approximation, Randomization, and Combinatorial Optimization.
  Algorithms and Techniques}, pages 646--657. Springer, 2009.

\bibitem[MORS10]{MORS10}
Kevin Matulef, Ryan O'Donnell, Ronitt Rubinfeld, and Rocco~A Servedio.
\newblock Testing halfspaces.
\newblock {\em SIAM Journal on Computing}, 39(5):2004--2047, 2010.

\bibitem[O'D14]{OD14}
Ryan O'Donnell.
\newblock {\em Analysis of Boolean Functions}.
\newblock Cambridge University Press, 2014.

\bibitem[OS11]{OS11}
Ryan O'Donnell and Rocco~A Servedio.
\newblock {The Chow Parameters Problem}.
\newblock {\em SIAM Journal on Computing}, 40(1):165--199, 2011.

\bibitem[Qi10]{Qi10}
Feng Qi.
\newblock Bounds for the ratio of two gamma functions.
\newblock {\em Journal of Inequalities and Applications}, 2010(1), 2010.

\bibitem[Ras03]{Ras03}
Sofya Raskhodnikova.
\newblock Approximate testing of visual properties.
\newblock In {\em Approximation, Randomization, and Combinatorial
  Optimization.. Algorithms and Techniques}, pages 370--381. Springer, 2003.

\bibitem[RS15]{RS15}
Dana Ron and Rocco~A Servedio.
\newblock Exponentially improved algorithms and lower bounds for testing signed
  majorities.
\newblock {\em Algorithmica}, 72(2):400--429, 2015.

\bibitem[SSBD14]{SB14}
Shai Shalev-Shwartz and Shai Ben-David.
\newblock {\em Understanding Machine Learning: From Theory to Algorithms}.
\newblock Cambridge University Press, 2014.

\bibitem[Wen48]{Wen48}
J.G. Wendel.
\newblock Note on the gamma function.
\newblock {\em The American Mathematical Monthly}, 55(9):563--564, 1948.

\bibitem[Win71]{Wind71}
Robert~O Winder.
\newblock Chow parameters in threshold logic.
\newblock {\em Journal of the ACM (JACM)}, 18(2):265--289, 1971.

\end{thebibliography}

\appendix

\section{Appendix}
\subsection{Proofs of the Preliminaries}
\textbf{Proposition \ref{prop:distance from mean}.}\textit{
Let $f,g$ be $\pm 1$-valued functions with $\abs{\Ex f} \geq 1-\epsilon$ and
$\abs{\Ex f - \Ex g} \leq \delta$. Then $\dist(f,g) \leq \epsilon + \delta/2$.
}
\begin{proof}
Suppose without loss of generality that $\Ex f \geq 1-\epsilon$. Then
$1-\epsilon \leq 1-2\Pr{f(x)=-1}$ so $\Pr{f(x)=-1} \leq \epsilon/2$. Now we have
$\delta \geq \abs{\Ex f - \Ex h} = 2\abs{\Pr{f(x) = -1} - \Pr{g(x) = -1}}$ so
\begin{align*}
  \dist(f,g)
  &= \Pr{f(x)=1,g(x)=-1} + \Pr{f(x)=-1,g(x)=1} \leq \Pr{g(x)=-1} + \Pr{f(x)=-1} \\
  &\leq 2\Pr{f(x)=-1} + \delta/2 \leq \epsilon + \delta/2 \,. \qedhere
\end{align*}
\end{proof}

\textbf{Proposition \ref{prop:isotropy}.}\textit{
Let $\mu$ be any RI distribution over $\bR^n$. $\mu$ is isotropic iff $\bE_{x
\sim \mu}[\norm{x}^2]=n$.}

\begin{proof}
By rotation invariance we have $\Ex{x_i^2} = \Ex{\inn{x,u}^2}$ for any unit
vector $u$. Then for any such $u$ we have
\[
  \Ex{\norm{x}^2} = \Ex{\inn{x,x}} = \sum_i \Ex{x_i^2} = n \Ex{\inn{x,u}^2} \,.
\]
If $\mu$ is isotropic this is $n$; in the opposite direction, if this quantity
is $n$ then $\Ex{\inn{x,u}^2}=1$ so $\mu$ is isotropic.
\end{proof}

We need the following proposition:
\begin{proposition}
\label{prop:sphere ratio}
  Let $S_n$ denote the surface area of the $n$-sphere. Then
  $\frac{\sqrt{n-2}}{\sqrt{2\pi}}
  \leq \frac{S_{n-2}}{S_{n-1}} \leq \frac{\sqrt{n-1}}{\sqrt{2\pi}}$.
\end{proposition}
\begin{proof}
  It is well known that $S_{n-1} = \frac{2\pi^{n/2}}{\Gamma\fracb{n}{2}}$, so the
  ratio is $ \frac{S_{n-2}}{S_{n-1}}
    = \frac{\Gamma\fracb{n}{2}}{\sqrt{\pi}\Gamma\fracb{n-1}{2}} $.
  The conclusion follows from applying the following inequality for $x>1,
  \epsilon \in [0,1]$:
  \[
    (x-1)^\epsilon \Gamma(x-\epsilon) \leq \Gamma(x) \leq (x-\epsilon)^\epsilon
    \Gamma(x-\epsilon)
  \]
  which can be found in \cite{Wen48, Qi10}.
\end{proof}

\textbf{Proposition \ref{prop:density of sphere}.}\textit{
Let $\sigma$ be the uniform distribution of the sphere of radius $r$ over
$\bR^n$. Let $\sigma_\pi$ be the density of the 1-dimensional projection.
Then $\sigma_\pi(x) \leq \frac{\sqrt{n-1}}{\sqrt{2\pi r^2}}
e^{-\frac{x^2(n-2)}{2r^2}}$, and in particular, for $r = \sqrt{n}$,
$\sigma_\pi(x) \leq \frac{1}{\sqrt{2\pi}} e^{-x^2/4}$ when $n \geq 4$.}
\begin{proof}
  Using the inequality $\frac{S_{n-2}}{S_{n-1}} \leq
  \frac{\sqrt{n-1}}{\sqrt{2\pi}}$, we get
  \begin{align*}
  \sigma_\pi(x)
  &= \frac{S_{n-2}(\sqrt{r-x^2})}{S_{n-1}(r)}
  = \frac{S_{n-2}}{S_{n-1}} \frac{(r-x^2)^{(n-2)/2}}{r^{n-1}}
  = \frac{S_{n-2}}{S_{n-1}} \frac{(1-(x/r)^2)^{(n-2)/2}}{r} \\
  &\leq \frac{S_{n-2}}{r S_{n-1}} e^{-\frac{x^2(n-2)}{2r^2}}
  \leq \frac{\sqrt{n-1}}{\sqrt{2\pi r^2}} e^{-\frac{x^2(n-2)}{2r^2}}
  \,. \qedhere
  \end{align*}
\end{proof}

\textbf{Proposition \ref{prop:center and normal are parallel}.}\textit{
For any RI distribution $\mu$ over $\bR^n$ and any halfspace $h(x) =
\sign(\inn{w,x}-t)$, there exists a scalar $s > 0$, such that $\bE_{x \sim
\mu}[x h(x)] = s w$.}
\begin{proof}
Without loss of generality assume that $w = e_1$.
Let $\sigma_r$ be the uniform distribution over the sphere of radius $r$ in
$\bR^n$ and pick any $\gamma$ such that $|\gamma| \leq r$. Let $x$ be drawn from
$\sigma_r$ conditioned on $x_1=\gamma$. Then $x_2\dotsm x_n$ is drawn from the
uniform distribution over the sphere with radius $\sqrt{r-\gamma^2}$ in
$\bR^{n-1}$, so for $i > 1, \Exuc{}{x_i}{x_1=\gamma} = 0$.  Thus $\bE_{x \sim
\sigma_r}[x \;\;\mid\;\; x_1=\gamma] = \gamma e_1$, and
\begin{align*}
  \Exu{x \sim \sigma_r}{ x\sign(x_1-t) }
  &= \int_{-r}^r \Exuc{}{x \sign(\gamma -t)}{x_1=\gamma } d\Pr{x_1=\gamma } \\
  &= \int_{-r}^r \sign(\gamma -t) \gamma e_1 d\Pr{x_1=\gamma }
  = \Exu{x \sim \sigma_r}{ x_1 \sign(x_1-t) } \cdot e_1 \,.
\end{align*}
Finally, recalling $w=e_1$, $\bE_{x \sim \mu}[xh(x)]
= \bE_{r \sim \mu_\circ}[\bE_{x \sim \sigma_r}[xh(x)]]
= \bE_{x \sim \mu}[ \inn{w,x} h(x) ] \cdot w$.
\end{proof}

\textbf{Proposition \ref{prop:distance decomposition}.}\textit{
  For any two halfspaces $g(x) = \sign(\inn{v,x}-q), h(x) = \sign(\inn{u,x}-p)$
  (where $u,v$ are unit vectors) and RI distribution $\mu$,
  \[
    \dist_\mu(g,h) \leq \frac{\alpha(g,h)}{\pi} + \hdist_\mu(p,q) \,.
  \]}

\begin{proof}
Without loss of generality, assume $|p| \geq |q|$, and let $h'$ be the halfspace
$h'(x) = \sign(\inn{v,x}-p)$. Then by the triangle inequality, $\dist(g,h) \leq
\dist(g,h') + \dist(h',h)$. Since $h',h$ have the same normal vector,
$\dist(h',h) = \hdist(p,q)$, so it remains to show that $\dist(g,h') \leq
\alpha(g,h)/\pi$. It suffices to prove this in the 2-dimensional case,
since we may project $\mu$ onto the span of $u,v$ to get another RI distribution
$\mu_\pi$ on which $\dist_{\mu_\pi}(g,h') = \dist_\mu(g,h')$. It further
suffices to consider the uniform distribution over the unit circle, since any RI
distribution in $\bR^2$ is a convex combination of scaled circles.

Consider $\sigma$, the uniform distribution over the sphere (circle) in
$\bR^2$ with radius $r$. Assume $0\leq p<1$, since if $p < 0$ we may consider
the halfspaces $-g, -h'$ with threshold $-p$ without changing their distance,
and if $p \geq 1$ then $g,h'$ are constant on $\sigma$ and the proposition
easily holds.

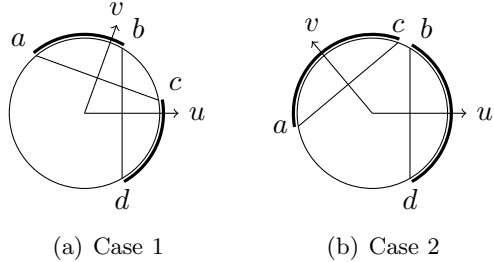
\begin{figure}[H]
\centering
\mbox{
\subfigure[Case 1]{
\begin{tikzpicture}
  \draw (0,0) circle [x radius=1, y radius=1];
  \draw[->] (0,0) -- (1.25,0) node[anchor=west]{$u$};
  \draw[->] (0,0) -- (70:1.25) node[anchor=south]{$v$};
  \draw (130:1)node[anchor=south east]{$a$} -- (10:1)node[anchor=south west]{$c$};
  \draw (-60:1)node[anchor=north]{$d$} -- (60:1)node[anchor=south west]{$b$} ;
  \draw[very thick] (130:1.05) arc (130:60:1.05) ;
  \draw[very thick] (10:1.05) arc (10:-60:1.05) ;
\end{tikzpicture}
}\quad
\subfigure[Case 2]{
\begin{tikzpicture}
  \draw (0,0) circle [x radius=1, y radius=1];
  \draw[->] (0,0) -- (1.25,0) node[anchor=west]{$u$};
  \draw[->] (0,0) -- (130:1.25) node[anchor=south]{$v$};
  \draw (190:1) node[anchor=east]{$a$}-- (70:1) node[anchor=south] {$c$};
  \draw (-60:1) node[anchor=north]{$d$} -- (60:1) node[anchor=south west]{$b$};
  \draw[very thick] (190:1.05) arc (190:70:1.05);
  \draw[very thick] (60:1.05) arc (60:-60:1.05);
\end{tikzpicture}
}
}
\caption{Proof of Proposition \ref{prop:distance decomposition}.}
\label{fig:distance decomposition}
\end{figure}

Write $\alpha = \alpha(g,h)$ for the angle between $u$ and $v$. Let $a,b,c,d \in
\bR^2$ be the points on the circle satisfying $\inn{a,v} = \inn{c,v} = \inn{b,u}
= \inn{d,u} = p$, $\inn{a,u} \leq \inn{c,u}, \inn{d,v} \leq \inn{b,v}$, and $a
\neq c, b \neq d$ (see Figure \ref{fig:distance decomposition}); from here we
conclude that $\alpha(a,v)=\alpha(c,v)=\alpha(b,u)=\alpha(d,u)$.

In case 1, where the halfspaces intersect, we can see that the difference (bold
in Figure \ref{fig:distance decomposition}) is bounded by
$(\alpha(a,b)+\alpha(c,d))/2\pi \leq \alpha/\pi$. And in case 2, where the
halfspaces are disjoint, the same relation holds with equality, which proves the
proposition.
\end{proof}

\textbf{Proposition \ref{prop:halfspace mean}.}\textit{
  Let $h$ be a halfspace with threshold $t$ and normal vector $w$. Then for
  any RI distribution,
  \[
  \abs{\Ex h} = \Pr{-|t| \leq \inn{w,x} \leq |t|} \,.
  \]}
\begin{proof}
  Without loss of generality, assume $t \leq 0$ (otherwise take $-h$).
  By rotation invariance we have $\Pr{\inn{w,x} < t} = \Pr{\inn{w,x} > -t}$.
  Then
  \begin{align*}
    \abs{\Ex h}
    &= \Pr{h(x)=1} - \Pr{h(x)=-1}
    = \Pr{\inn{w,x} \geq t} - \Pr{\inn{w,x} < t} \\
    &= \Pr{\inn{w,x} \geq t} - \Pr{\inn{w,x} > -t}
    = \Pr{t \leq \inn{w,x} \leq t} \,. \qedhere
  \end{align*}
\end{proof}

\ignore{
\begin{proof}[Proof of Proposition \ref{prop:sphere absolute value}]
  For radius 1,
  \begin{align*}
    2\int_0^1 x \frac{S_{n-2}(\sqrt{1-x^2})}{S_{n-1}(r)} dx
    &= 2 \frac{S_{n-2}}{S_{n-1}} \int_0^r x (1-x^2)^{n-2} dx \\
    &\leq 2 \frac{S_{n-2}}{S_{n-1}} \int_0^r x e^{-x^2(n-2)} dx \\
    t = x^2(n-2), dt = 2x(n-2): \qquad \qquad
    &\leq \frac{S_{n-2}}{(n-2)S_{n-1} r^{n-1}} \int_0^\infty e^{-t} dt \\
    &\leq \frac{\sqrt{n-1}}{\sqrt{2\pi}(n-2)}
    \leq \frac{1}{\sqrt{\pi(n-2)}} \xh{(since $\sqrt{(n-1)/(n-2)} \leq \sqrt{2}$)} \\
    &\leq \frac{\sqrt{2}}{\sqrt{\pi n}} \xh{(for $n \geq 4$)}.
  \end{align*}
  For the lower bound, observe that the maximum density of the projected
  distribution is $\mu_{max} = \frac{S_{n-2}}{S_{n-1}} \leq \sqrt{n/2\pi}$
  (Proposition ??). Let $t \geq 0$ such that $\Pr{\abs{\inn{x,u}} < t} = 1/2$.
  Then $t\sqrt{n/2\pi} \geq 1/2$ so $t \geq \sqrt{\pi/2n}$. Thus
  \[
  \Ex{\abs{\inn{x,u}}} \geq \frac{1}{2}t \geq \frac{\sqrt{\pi}}{\sqrt{8n}} \,.
  \]
\end{proof}
}

\textbf{Proposition \ref{prop:tail bounds for spheres}.}\textit{
Let $\sigma$ be the uniform distribution over the sphere of radius $r$, and let
$t \geq 0$. Then for any unit vector $u$,
\[
\Pru{x \sim \sigma}{\inn{u,x} \geq t} \leq \sqrt{2}e^{-t^2(n-2)/2r^2} 
\leq \sqrt{2}e^{-t^2n/4r^2}
\]
(where the last inequality holds when $n \geq 4$).}
\begin{proof}
Suppose $\sigma$ is the unit sphere and without loss of generality assume
$u=e_1$.
\begin{align*}
\Pr{x_1 \geq t}
&= \frac{S_{n-2}}{S_{n-1}} \int_t^1 (1-t^2)^{(n-2)/2} dt
\leq \frac{S_{n-2}}{S_{n-1}} \int_t^\infty e^{-t^2(n-2)/2} dt \\
&= \frac{S_{n-2}}{S_{n-1}} \frac{\sqrt{2\pi}}{\sqrt{n-2}} \cdot
  \Pru{z \sim \cN(0, 1/\sqrt{n-2})}{z \geq t}
  \leq \frac{S_{n-2}}{S_{n-1}} \frac{\sqrt{2\pi}}{\sqrt{n-2}} e^{-t^2(n-2)/2}
\end{align*}
where the final inequality for the Gaussian distribution can be found, for
example, in \cite{BLM13}. The conclusion then follows from an application of
Proposition \ref{prop:sphere ratio}.
\end{proof}

\subsection{Concentration Inequalities and Empirical Estimation}
\label{appendix:concentration inequalities}
We will make use of two standard concentration inequalities:
\begin{lemma}[Chebyshev's Inequality]
Let $X$ be any random variable over $\bR$ and let $t \geq 0$. Then $\Pr{X \geq
t} \leq \frac{\Var{X}}{t^2}$.
\end{lemma}
\begin{lemma}[Hoeffding's Inequality]
Let $\{X_i\}_{i \in [m]}$ be a set of $m$ independent random variables over
$\bR$ and let $\{a_i,b_i\}_{i \in [m]}$ be pairs of numbers such that for each
$i \in [m], \Pr{a_i \leq X_i \leq b_i}=1$. Then for $X = \sum_i X_i$, and any $t
\geq 0$,
$\Pr{ X - \Ex X > t} \leq \exp{-\frac{2t^2}{\sum_i (b_i-a_i)^2}}$ and the same
bound holds for $\Pr{ \Ex X - X > t }$.
\end{lemma}
We will use standard empirical estimation techniques, which are easy
consequences of Hoeffding's Inequality. For clarity of presentation, we
present this as a subroutine:
\begin{algorithm}[H]
\label{alg:estimate mean}
\caption{\textsc{Estimate-Mean}($\mu, f, \epsilon, \delta$)}
\begin{algorithmic}[1]
\State $m \gets \frac{2}{\epsilon^2}\ln(2/\delta)$
\State Let $(x_i, f(x_i))_{i \in [m]}$ be a set of random samples.
\State \Return $\frac{1}{m} \sum_{i \in [m]} f(x_i)$
\end{algorithmic}
\end{algorithm}
\begin{lemma}
\label{lemma:estimate mean}
For any distribution $\mu$ over $\bR^n$ and function $f : \bR^n \to R$,
algorithm \textsc{Estimate-Mean} requires at most
$\BigO{\frac{1}{\epsilon^2}\log(1/\delta)}$ random samples and will produce an
estimate $\tilde E$ such that, with probability at least $1-\delta$,
$\abs{\tilde E - \bE_{\mu}{f}} < \epsilon$.
\end{lemma}
\begin{proof}
  By Hoeffding's inequality, $\Pr{\abs{\tilde E - \Ex{f}} \geq \epsilon}
  \leq 2\exp{-m\epsilon^2/2} = \delta$ when $m =
  \frac{2}{\epsilon^2}\ln(2/\delta)$.
\end{proof}

The next proposition gives a guarantee on a simple method for estimating
thresholds:
\begin{proposition}
\label{prop:estimate threshold}
  Let $\epsilon,\delta > 0$ and let $\mu$ be some distribution over $\bR$.
  For any $v \in [0,1]$ fix $t$ such that $\Pr{x \geq t}=v$ and let $X$ be a
  set of $m = \frac{1}{2\epsilon^2}\ln(2/\delta)$ random samples. Let $p =
  \max\{z : \#\{x \in X : x \geq z\} \geq mv \}$.  If $p \leq t$ let $I =
  [p,t)$, otherwise let $I = [t,p)$. Then
  \[
  \Pru{X}{ \mu(I) > \epsilon } \leq \delta \,.
  \]
\end{proposition}
\begin{proof}
  Let $a = \max\{z : \mu[z,t) \geq \epsilon \}, b = \min\{z : \mu[t,z) \geq
  \epsilon \}$, or $a=-\infty, b=\infty$ if no such values exist. The
  estimation fails if $p \leq a$ or $p \geq b$; note that if $a=-\infty$ or
  $b=\infty$ the respective inequalities cannot be satisfied, so we are
  concerned only with the cases where $a,b$ are finite. If $p \leq a$ then $A
  \define \#\{x \in X : x > a \} < mv$, which, since $\Ex{A/m} \geq v +
  \epsilon$, happens with probability
  \[
  \Pr{ A/m < v } \leq \Pr{ A/m < \Ex{A/m} - \epsilon } \leq \exp{-2m\epsilon^2}
  \]
  by Hoeffding's inequality. The same holds for $B \define \#\{x \in X : x
  \geq b\}$. Thus, taking the union bound and setting $m =
  \frac{1}{2\epsilon^2}\ln(2/\delta)$ we see that the failure probability is
  at most $\delta$.
\end{proof}

\end{document}